\documentclass{article}
\usepackage{graphicx}

\usepackage[utf8]{inputenc}
\usepackage[english]{babel}

\usepackage{relsize}
\usepackage{amsfonts,amsmath,amssymb,amsbsy,amsthm,stmaryrd,enumerate}
\usepackage[autostyle]{csquotes}

\usepackage[shortlabels]{enumitem}
\usepackage{hyperref}
\usepackage{comment}
\usepackage{placeins}

\usepackage{tikz}
\usetikzlibrary{patterns}
\usepgflibrary{decorations.pathmorphing}
\usetikzlibrary{decorations.pathmorphing}
\usetikzlibrary{calc,shapes.arrows,decorations.pathreplacing}
\usetikzlibrary{shadows}
\usetikzlibrary{positioning}

\usepackage[ruled]{algorithm2e}

\usepackage{subcaption}
\usepackage{float}

\usepackage{setspace}

\usepackage{cleveref}

\newtheorem{property}{Property}
\newtheorem{thm}{Theorem}
\newtheorem{Lemma}{Lemma}

\newtheorem{lemma}[Lemma]{Lemma}

\newtheorem{definition}{Definition}
\newtheorem{notation}{Notation}

\newtheorem{proposition}{Proposition}
\newtheorem{fact}{Fact}
\newtheorem{remark}{Remark}

\theoremstyle{remark}

\theoremstyle{definition}

\usepackage[noend]{algpseudocode}
\usepackage{multicol}
\usepackage{fullpage}

\makeatletter
\def\BState{\State\hskip-\ALG@thistlm}
\makeatother

\tikzstyle{mydashed}=[dash pattern=on 1.5pt off 2pt]

\begin{document}
\title{Asymptotic growth rate of square 
grids dominating sets: a symbolic dynamics approach}

\author{Silvère Gangloff \and Alexandre Talon}
\date{}
\maketitle
\begin{center}
ENSL, Univ Lyon, UCBL, LIP, MC2, F-69342, LYON Cedex 07, France\\
\texttt{silvere.gangloff@ens-lyon.fr},   \texttt{alexandre.talon@ens-lyon.fr}
\end{center}
\begin{abstract}
In this text, we prove 
the existence of an asymptotic growth rate 
of the number of dominating sets (and variants)
on finite rectangular grids, when the dimensions of the 
grid grow to infinity.
Moreover, we provide, for each of the variants, an algorithm
which computes the growth rate.
We also give bounds on these rates provided by a computer program.
\end{abstract}

\section{Introduction}

A \textit{dominating set} $S$ of a graph $G$ is 
a subset of its vertices such that 
any vertex not in $S$ is connected to 
a vertex in $S$.
The dominating number $\gamma(G)$ 
of the graph 
is the minimum cardinality of a 
dominating set.
These notions appear in practical 
problems, related to robotics and 
networks constructions.
Some decidability results are known: 
for instance, the problem $\gamma(G) \le k$ 
given the integer $k$ and the finite graph $G$ is NP-complete.
An important problem has been to compute exactly the dominating number for finite 
rectangular grids, and it was solved by D. Gon\c calves, 
A. Pinlou, M. Rao, S. Thomass\'e~\cite{rao}, proving 
Chang's conjecture, which tells 
that denoting $G_{n,m}$ the finite $n \times m$ rectangular grid, 
\[\gamma(G_{n,m}) = \Bigl \lfloor \frac{(n+2)(m+2)}{5}\Bigr \rfloor -4.\]
Another problem, which is still open, 
is to compute the number of dominating set 
of graphs. Some formulas are known, such as a relation between this number and the number of complete bipartite subgraphs of the complement of $G$~\cite{Heinrich}. \bigskip

In the present text, we are interested in the asymptotic growth rate of the number of dominating sets, on the finite rectangular grids $G_{n,m}$, when $n$ and $m$ grow to infinity. We also study this problem for the total domination, the minimal domination, and the minimal total domination. The text is organised as follows:

\begin{enumerate}

\item In Section~\ref{notions.dominating.sets}, we define the various notions 
of dominating sets on (finite or infinite) graphs we have just mentioned, 
and prove local characterisations of 
these sets. 

\item In Section~\ref{section.subshifts.finite.type} we associate, to each of these notions of dominating sets, a symbolic dynamical system called subshift of finite type, which consists in a set of colourings of the infinite grid $\mathbb{Z}^2$. Comparing the number of dominating sets on finite grids and the number of patterns which appear in configurations of the corresponding subshift, we prove the existence of a growth rate and show that it is equal to the entropy of the dynamical system.
% ^ domination problems 
\item In Section~\ref{section.block.gluing.property}, we define the block-gluing property; any subshift of finite type that is block-gluing is guaranteed to have an entropy which is computable in an algorithmical sense. We then prove that the various domination subshifts defined in Section~\ref{section.subshifts.finite.type} are block gluing. 
This fact provides an algorithm which computes approximations of the growth rate given the desired precision in the input.
\item In Section~\ref{section.computer}, we provide some bounds for the growth rates obtained by a computer program.
\end{enumerate}

\section{\label{notions.dominating.sets} Notions of dominating sets of square grids}

\subsection{\label{section.definitions.dominating}Definitions}

In the following, for a graph $G=(V,E)$, we will say that two vertices $u,v$ in $V$ are \textit{neighbours} or \textit{connected} when the edge $\{u,v\}$ is in $E$. For all $n \ge 1$, we will denote by $G_{n,m}$ the finite square grid graph of size $n \times m$.

\begin{definition}
Let $G=(V,E)$ be a graph, $S$ a subset of $V$ and $v \in V$. We say that $v$ is \textbf{dominated} by $S$ if $v$ has a neighbour in $S$.
\end{definition}

\begin{definition}
\label{definition.dominating.sets}
Let $G=(V,E)$ be a graph. A subset $S \subseteq V$ is said to be a \textbf{dominating set} of the graph $G$ when every $v \in V \setminus S$ is dominated by $S$. It is said to be a \textbf{minimal dominating set} of $G$ when it is a dominating set of $G$ and for all $v \in S$, $S \setminus \{v\}$ is not dominating. It is said \textbf{total dominating} when for all $v$ in $V$, $v$ has a neighbour in $S$.
\end{definition}

\begin{definition}
\label{definition.derived.domination.notions}
A subset $S \subset V$ is said to be \textbf{minimal total dominating} when it is total dominating and for all $v \in S$, $S \setminus \{v\}$ is not total dominating.
\end{definition}

Notice that a minimal dominating set (resp. minimal total dominating set) is a dominating set (resp. total dominating set) which is inclusion-wise minimal. These notions are illustrated in 
Figure~\ref{figure.domination.notions}.

\begin{figure}[h!]
\centering
\begin{tikzpicture}[scale=0.5]
\begin{scope}

\fill[gray!85] (0,0) rectangle (2,2);
\fill[gray!85] (1,3) rectangle (2,4);
\fill[gray!85] (2,2) rectangle (3,3);
\fill[gray!85] (3,3) rectangle (4,4);
\fill[gray!85] (3,1) rectangle (4,2);
\draw (0,0) grid (4,4);
\node at (1.5,-1) {(a)};
\node at (-0.8,0) {\scriptsize (1,1)};
\end{scope}

\begin{scope}[xshift=6cm]

\fill[gray!85] (1,3) rectangle (2,4);
\fill[gray!85] (0,1) rectangle (1,2);
\fill[gray!85] (1,0) rectangle (2,1);
\fill[gray!85] (3,1) rectangle (4,3);

\draw (0,0) grid (4,4);
\node at (1.5,-1) {(b)};
\end{scope}

\begin{scope}[xshift=12cm]

\fill[gray!85] (0,1) rectangle (2,2);
\fill[gray!85] (1,0) rectangle (2,1);
\fill[gray!85] (1,3) rectangle (3,4);
\fill[gray!85] (3,1) rectangle (4,3);

\draw (0,0) grid (4,4);

\node at (1.5,-1) {(c)};
\end{scope}

\end{tikzpicture}
\caption{\label{figure.domination.notions} Illustration of Definition~\ref{definition.dominating.sets} and 
Definition~\ref{definition.derived.domination.notions} on $G_{4,4}$:\\
(a) a dominating set which is neither minimal dominating nor total dominating;\\
(b) a minimal dominating set which is not total dominating (the bottom-left dominant vertices are not dominated);\\
(c) a minimal total dominating set.}
\end{figure}

\begin{definition}
\label{definition.local.domination.notions}
When a dominating set $S$ of a graph $G$ is fixed, a vertex is called a \textbf{dominant} element of $G$ when it is in $S$, and a \textbf{dominated} element when it has a neighbour in $S$. A neighbour $w$ of a dominant element $v$ is said to be a \textbf{private neighbour} of $v$ when $v$ is the only neighbour of $w$ in the set $S$.
\end{definition}

\subsection{\label{section.local.characterisations} Local characterisations}

In this section, we recall, and for completeness 
prove, local characterisations 
of the notions of dominating sets. This means that 
one can check if a set $S$ is dominating (or minimal dominating, etc) by checking, for 
each vertex, whether or not this vertex and its neighbours 
are in $S$

\begin{fact}
\label{fact.dominating.sets}
Let $S$ be a set of vertices of a graph $G=(V,E)$. 
Then for all $v$ and $w$ such that $w$ is not a neighbour of $v$, $w$ is dominated by $S$ if and only if it is dominated by $S \setminus \{v\}$.
\end{fact}

\begin{definition}
Let $S$ be a set of vertices of a graph $G=(V,E)$. 
We say that a dominant element is \textbf{isolated} in $S$ when it has no neighbours 
in $S$.
\end{definition}

\begin{proposition}
Let $S$ be a dominating set of a graph $G = (V,E)$. $S$ is minimal dominating if and only if any of its elements is isolated in $S$ or has a private neighbour not in $S$. 
\end{proposition}

\vspace{5cm}
\begin{proof} \leavevmode
\begin{itemize}
	\item $(\Rightarrow)$: Let us assume that $S$ is minimal dominating, and fix $v \in S$. 	From Fact~\ref{fact.dominating.sets} and by definition of a minimal dominating set, any $w$ which is not in the neighbourhood of $v$ is dominated by $S \backslash \{v\}$. Since $S \setminus \{v\}$ is not dominating, it means that: 
	\begin{enumerate}
		\item $v$ is not dominated by $S \setminus \{v\}$, which means that $v$ is isolated in $S$,
		\item or there exists some $u \notin S$ neighbour of $v$ which is not dominated by $S \setminus \{v\}$, hence $u$ is a private neighbour of $v$ which is not in $S$.
		
	\end{enumerate}
%	The fact that every vertex not in $S$ has a neighbour in $S$ comes directly from the fact that $S$ is a dominating set.
	\item $(\Leftarrow)$: Conversely, let us fix some dominating set $S$ such that every $v \in S$ has a private neighbour not in $S$ or is isolated. Fix some $v \in S$. If it has a private neighbour $u$, then $u$ is not dominated by $S \setminus \{v\}$, and thus $S \setminus \{v\}$ is not dominating. If it has no private neighbours, then it is isolated. This means that $v$ is not dominated by $S \setminus \{v\}$, therefore the set is not dominating. In both cases, we conclude that $S$ is minimal dominating.
	\end{itemize}
\end{proof}

With a similar proof, we obtain the following:

\begin{proposition}
A total dominating set $S$ of a graph $G$ is minimal total dominating if and 
only if any $v \in S$ has a private neighbour.% and is not isolated. PAS BESOIN %TODO
\end{proposition} 

In the following, we will use the following notations:

\begin{notation}
In the following, for all integers $n,m$, we denote by $D_{n,m}$, $M_{n,m}$ and $T_{n,m}$ respectively the number of dominating sets of the grid $G_{n,m}$, the number of its minimal dominating sets and the number of its minimal total dominating sets.
\end{notation}

\section{\label{section.subshifts.finite.type} From dominating sets to subshifts of finite type}

In this section, we introduce the notion of subshift of finite type on a regular grid (see Section~\ref{section.sft.definition}), which consists in sets of possible colourings of the grid avoiding some forbidden patterns. After presenting some examples which are the subshifts counterparts of various notions 
of domination in Section~\ref{section.examples}, we use the 
well-known fact that the entropy of a subshift can be 
expressed as a limit to prove the existence of an asymptotic
growth rate of the number of dominating sets in Section~\ref{section.proof.asymptotic}.

\subsection{\label{section.sft.definition} Subshifts of finite type}

\begin{definition}
Let $\mathcal{A}$ be a finite set, and $d \ge 1$ integer.
A \textbf{pattern} $p$ on alphabet $\mathcal{A}$ is 
an element of $\mathcal{A}^{\mathbb{U}}$ for some finite
$\mathbb{U} \subset \mathbb{Z}^d$. The set $\mathbb{U}$ 
is called the \textbf{support} of $p$, and is denoted $\text{supp}(p)$. Informally, 
it is the location of $p$ in the grid $\mathbb{Z}^d$.
\end{definition}

\begin{notation}
For a configuration $x = (x_\textbf{u})_{\textbf{u} \in \mathbb{Z}^d}$ of $\mathcal{A}^{\mathbb{Z}^d}$ 
(resp. a pattern $p \in \mathcal{A}^{\mathbb{U}}$ for some 
$\mathbb{U} \subset \mathbb{Z}^d$), we denote by 
$x_{|\mathbb{V}}$ the restriction of $x$ 
to some subset $\mathbb{V} \subset \mathbb{Z}^d$ 
(resp. the restriction of $p$ to $\mathbb{V} \subset \mathbb{U}$).
\end{notation}

\begin{definition}
Let $\mathcal{A}$ be a finite set, and $d \ge 1$ integer.
A \textbf{$d$-dimensional subshift of finite type} (SFT) on alphabet $\mathcal{A}$ 
is a subset of $\mathcal{A}^{\mathbb{Z}^d}$
defined by a finite set of forbidden patterns. Formally, a subset 
$X$ of $\mathcal{A}^{\mathbb{Z}^d}$ is a subshift 
of finite type when there exist some finite 
sets $\mathbb{U} \subset \mathbb{Z}^d$ and 
$\mathcal{F} \subset \mathcal{A}^{\mathbb{U}}$ such that: 
\[X= \left\{x \in \mathcal{A}^{\mathbb{Z}^d}: 
\forall\, \textbf{u} \in \mathbb{Z}^d, x_{|\textbf{u}+\mathbb{U}} \notin \mathcal{F}\right\}.\]
The elements of $\mathcal{F}$ are called the \textbf{forbidden patterns}.
When the set of forbidden patterns is fixed, 
a pattern which does not contain any forbidden 
pattern is called \textbf{locally admissible}.
\end{definition}

\begin{notation}
Let us denote by $\sigma$ the $\mathbb{Z}^d$-\textbf{shift} action on $\mathcal{A}^{\mathbb{Z}^d}$ defined such that for all $\textbf{u},\textbf{v} \in \mathbb{Z}^d$, 
\[(\sigma^{\textbf{u}} x)_{\textbf{v}} = x_{\textbf{v}+\textbf{u}}\texttt{.}\]
Informally, $\sigma$ acts on a configuration by translating it by the vector $\textbf{u}$.
\end{notation}

\begin{definition}
For a SFT $X$ a \textbf{globally admissible} pattern of size $\llbracket 1, n \rrbracket ^d$ is some $p \in \mathcal{A}^{\llbracket 1,n\rrbracket ^d}$ which appears in a configuration of $X$, that is when $x_{|\llbracket 1,n\rrbracket ^d}=p$.
When $d = 2$, we extend the definition to patterns $p \in \mathcal{A}^{\llbracket 1,n\rrbracket \times \llbracket 1, m \rrbracket}$ when there exists a configuration $x$ of $X$ such that $x_{|\llbracket 1,n\rrbracket \times \llbracket 1, m \rrbracket}=p$.
\end{definition}

\begin{notation}
For a subshift of finite type $X$, we denote by
$N_n (X)$
the number of globally admissible patterns of size $\llbracket 1,n\rrbracket ^d$. When $d=2$, we extend the notation and denote by resp. $N_{n,m}(X)$ the number of globally admissible patterns of size $\llbracket 1,n \rrbracket \times \llbracket 1,m \rrbracket$.
\end{notation}

\begin{definition}
The \textbf{topological entropy} of a subshift 
of finite type is the number
\[h(X) = \inf_n \frac{\log_2 (N_n (X))}{n^d}.\]
\end{definition}

The following three lemmas are well known (see for instance~\cite{Lind-Marcus}).

\begin{lemma}
The infimum in the definition of $h(X)$ is in fact a limit: 
\[h(X) = \lim_n \frac{\log_2 (N_n (X))}{n^d}.\]
\end{lemma}

\begin{definition}
A \textbf{conjugation} between two $d$-dimensional subshifts of finite type $X$ and $Z$ is an invertible 
map $\varphi: X \rightarrow Z$ such that 
for all $\textbf{u} \in \mathbb{Z}^d$ and $x \in X$,
$\varphi(\sigma^{\textbf{u}}.x) = \sigma^{\textbf{u}}.\varphi (x)$.
In this case, $X$ and $Z$ are said to be \textbf{conjugated}.
\end{definition}

\begin{lemma}
If two subshifts of finite type $X$ and $Z$ 
are conjugated, then $h(X)=h(Z)$. 
\end{lemma}

\begin{lemma}
Let $X$ be a bidimensional subshift of finite type. Then:
\[h(X) = \lim_{n,m} \frac{\log_2 (N_{n,m} (X))}{nm}.\]
\end{lemma}

\subsection{\label{section.examples} Domination subshifts}

In this section, the alphabet is 
$\mathcal{A}_0 =\left\{\begin{tikzpicture}[scale=0.3]
\draw (0,0) rectangle (1,1);
\end{tikzpicture},\begin{tikzpicture}[scale=0.3]
\fill[gray!85] (0,0) rectangle (1,1);
\draw (0,0) rectangle (1,1);
\end{tikzpicture}\right\}$, and $d=2$. 

\begin{definition}
The \textbf{domination} (resp. \textbf{minimal domination}, \textbf{total domination} and \textbf{minimal total-domination}) denoted by $X^{D}$ (resp. $X^{M}$, $X^T$ and $X^{MT}$), is 
the set of elements $x$ of $\mathcal{A}_0^{\mathbb{Z}^2}$ such that 
$\{\textbf{u} \in \mathbb{Z}^2: x_{\textbf{u}} = \begin{tikzpicture}[scale=0.3]
\fill[gray!85] (0,0) rectangle (1,1);
\draw (0,0) rectangle (1,1);
\end{tikzpicture}\}$ is 
a dominating (resp. minimal dominating, 
total dominating and minimal total dominating) 
set of the infinite square grid $\mathbb{Z}^2$.
In all these cases, a configuration 
$x$ of the subshift is called a \textbf{dominated 
configuration}. We also say that $\textbf{u}$ is a \textbf{dominant position} of the configuration $x$ when $x_\textbf{u}$ is grey. Likewise, a private neighbour is still a position which is dominated by exactly one dominant position.
\end{definition}

\begin{property}
The sets $X^{D}$, $X^{M}$, $X^T$ and $X^{MT}$
are subshifts of finite type.
\end{property}

\begin{proof}
This comes from the local characterisations of each type of dominating sets [Section~\ref{section.local.characterisations}], 
which can straightforwardly be translated into 
forbidden patterns. 
\end{proof}

%-------------------------------TODO from here thèse-------------------------------

\subsection{\label{section.proof.asymptotic} 
Existence of an asymptotic growth rate}

The dominating 
sets of a finite grid $G_{n,m}$ do not correspond 
exactly to the globally admissible 
patterns on the same grid of the corresponding subshifts of finite type presented in Section~\ref{section.examples}.
Indeed, in such a pattern, the positions of the border 
may for instance be dominated by a position outside the pattern in a configuration 
in which it appears. However we will see that we can compare the number of globally admissible patterns of size $n \times m$ for $X^D$ (resp. $X^M$, $X^T$ and $X^{MT}$) with the number of dominating sets (resp. minimal dominating sets, total dominating set and minimal total dominating sets) of $G_{n,m}$. We use this to prove the existence of an asymptotic growth rate for the grid, equal to the entropy of the SFT.

In this section,
we assimilate the set of vertices of $G_{n,m}$ to any translate of 
$\llbracket 1,n \rrbracket \times \llbracket 1,m \rrbracket$ and
assimilate any set $S$ of vertices of a finite grid $G_{n,m}$ with the pattern $p$ of $\mathcal{A}^{\mathbb{Z}^2}$ on $\llbracket 1,n \rrbracket \times \llbracket 1,m \rrbracket$
defined by $p_{\textbf{u}}$ being grey if and only 
if $\textbf{u} \in S$.

\begin{notation}
If $\mathbb{U}$ is a subset of $\mathbb{Z}^2$, 
we define the (extended) \textbf{neighbourhood of $\mathbb{U}$} as
\[\mathcal{N} (\mathbb{U}) = \bigcup_{\textbf{u} \in \mathbb{U}} \left( \textbf{u} + \llbracket -1,1\rrbracket ^2\right).\]
We also define, for all $n,m \ge 1$ and $k \geq 1$ the \textbf{border}
\[\mathbb{B}_{n,m,k} = \mathcal{N}^{k} (\llbracket 1,n \rrbracket \times \llbracket 1,m \rrbracket) \setminus \mathcal{N}^{k-1} (\llbracket 1,n \rrbracket \times \llbracket 1,m \rrbracket).\] 

\[\begin{tikzpicture}[scale=0.2]

\fill[gray!50] (-3,-3) rectangle (12,7);
\fill[pattern=north west lines] (-2,-2) rectangle (11,6);
\draw (-2,-2) rectangle (11,6);
\draw (-3,-3) rectangle (12,7);
\node at (17.5,2) {$\mathbb{B}_{m,n,k+1}$};
\draw[-latex] (13.5,2) -- (11.5,2);

\draw[fill=gray!30] (0,0) rectangle (9,4);
\draw[fill=black] (0,0) circle (4pt) ;

\draw[-latex] (-4.5,0) -- (-0.4,0);
\node[scale=0.7] at (-6.2,0) {$(1,1)$};

\draw[-latex] (9,-4.5) -- (4.5,-4.5) -- (4.5,-1);
\node at (17,-4.5) {$\mathcal{N}^k (\llbracket 1,n \rrbracket \times 
\llbracket 1,m \rrbracket)$};

\node[scale=0.7] at (4.5,2) {$\llbracket 1,n \rrbracket \times \llbracket 1,m \rrbracket$};

\end{tikzpicture}\]
For convenience, we extend the notation to $\mathbb{B}_{n,m,0}
= \llbracket 1,n \rrbracket \times \llbracket 1,m \rrbracket \setminus \llbracket 2,n-1 \rrbracket \times \llbracket 2,m-1 \rrbracket$.
\end{notation}

\begin{lemma}
For all $n,m \ge 2$, the following inequalities hold: $N_{n-1,m-1} (X^D) \le D_{n,m} \le N_{n,m} (X^D)$.
\end{lemma}
%TODO inégalité avec total => la même?
\begin{proof} \leavevmode
\begin{enumerate}
\item For all $n,m \ge 1$, any dominating set of $G_{n,m}$ 
can be extended into a configuration of $X^D$ 
by defining the symbol of 
any position outside $\llbracket 1,n \rrbracket \times \llbracket 1,m \rrbracket$ to be grey. As a consequence, 
any dominating set of $G_{n,m}$ is globally 
admissible in $X^D$ and thus $D_{n,m} \le N_{n,m} (X^D)$.
\item Any pattern of $X^D$ on $\llbracket 1,n\rrbracket \times \llbracket 1,m \rrbracket$ can be turned into 
a dominating set of $\llbracket 0,n+1\rrbracket
\times \llbracket 0,m+1 \rrbracket$
by extending it with grey symbols. Hence we obtain
the inequality $N_{n,m} (X^D) \le D_{n+1,m+1}$ 
for all $n,m \ge 1$. 
\end{enumerate}
\end{proof}

Using the very same arguments, we obtain the same inequality for the total domination.
\begin{lemma}
For all $n,m \ge 2$, the following inequalities hold: $N_{n-1,m-1} (X^T) \le T_{n,m} \le N_{n,m} (X^T)$.
\end{lemma}

We then address the minimal and minimal total domination. The proofs of the following inequalities are more complex as we will see.
\begin{lemma}
\label{lemma.comparison.minimal}
For all $n,m \ge 1$, the following inequalities hold: 
\[\frac{1}{2^{6(n+m)}} N_{n,m} (X^M) \le M_{n,m} \le N_{n,m} (X^M).\]
\end{lemma}

\begin{proof} \leavevmode
\begin{enumerate}
\item \textbf{Second inequality.} 

\begin{enumerate}
\item \textbf{A completion algorithm 
of a dominating set into a configuration of 
$X^M$.} 
Let $S$ be a minimal dominating set of $\llbracket 1,n\rrbracket \times \llbracket 1,m \rrbracket$. Let us extend it into a configuration $x$ of $X^M$ using the following algorithm: successively for every 
$k \ge 0$, we extend the current pattern 
into a pattern on $\mathcal{N}^{k+1} (\llbracket 1,n \rrbracket \times \llbracket 1,m \rrbracket)$ using the following operations, for all $\textbf{u} \in \mathbb{B}_{n,m,k+1}$:

\begin{enumerate}
	\item if $\textbf{u}$ is a corner then $x_{\textbf{u}}$ is white;
	\item if $\textbf{u}$ is a neighbour of a corner in one of the vertical sides of $\mathbb{B}_{n,m,k+1}$ then $x_{\textbf{u}}$ is white;
	\item for all other $\textbf{u}$, $x_{\textbf{u}}$ is grey if and only if its neighbour in $\mathcal{N}^k (\llbracket 1,n \rrbracket \times \llbracket 1,m \rrbracket) $ is neither dominated by an element in this set, nor a dominant element.
\end{enumerate}

This algorithm is illustrated in 
Figure~\ref{figure.completion.crowns.minimal}.

\begin{figure}[h!]
\[\begin{tikzpicture}[scale=0.25]
\fill[gray!10] (0,0) rectangle (5,5);
\draw (0,0) rectangle (5,5);

\begin{scope}[xshift=12cm]
\fill[gray!10] (0,0) rectangle (5,5);
\draw (0,0) rectangle (5,5);
\draw (5,5) rectangle (6,6);
\draw (0,0) rectangle (-1,-1);
\draw (0,5) rectangle (-1,6);
\draw (5,0) rectangle (6,-1);

\draw (5,5) rectangle (6,4);
\draw (5,0) rectangle (6,1);
\draw (0,0) rectangle (-1,1);
\draw (0,5) rectangle (-1,4);

\end{scope}

\begin{scope}[xshift=24cm]
\fill[gray!10] (0,0) rectangle (5,5);
\fill[white] (1,4) rectangle (4,5);
\fill[white] (2,3) rectangle (3,4);
\draw (1,4) grid (4,5);
\draw (2,3) rectangle (3,4);
\draw[fill=gray!85] (2,5) rectangle (3,6);

\fill[white] (1,0) rectangle (3,1);
\fill[gray!85] (3,0) rectangle (4,1);
\fill[white] (2,1) rectangle (3,2);
\draw (1,0) grid (4,1);
\draw (2,1) rectangle (3,2);
\draw (2,-1) rectangle (3,0);

\draw (0,0) rectangle (5,5);
\draw (5,5) rectangle (6,6);
\draw (0,0) rectangle (-1,-1);
\draw (0,5) rectangle (-1,6);
\draw (5,0) rectangle (6,-1);

\draw (5,5) rectangle (6,4);
\draw (5,0) rectangle (6,1);
\draw (0,0) rectangle (-1,1);
\draw (0,5) rectangle (-1,4);
\end{scope}
\end{tikzpicture}\]
\caption{\label{figure.completion.crowns.minimal} Illustration of the completion algorithm in $X^M$: steps of the algorithm are applied successively from left to right.}
\end{figure}
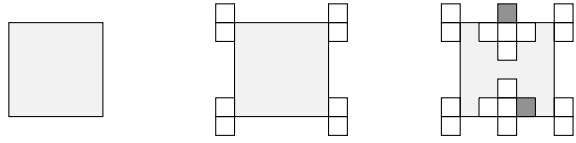

\item \textbf{The output obtained by repeating the algorithm 
is a configuration of $X^M$.}

\begin{itemize}
\item \textbf{Every position is dominated.}
This is verified for the positions in
$\llbracket 2,n-1\rrbracket \times \llbracket 2,m-1\rrbracket$. Outside 
this set, if a position in some $\mathcal{N}^k (\llbracket 1,n \rrbracket \times \llbracket 1,m \rrbracket)$ (for $k \ge 0$) is not dominated 
before extending the configuration 
on $\mathcal{N}^{k+1} (\llbracket 1,n \rrbracket \times \llbracket 1,m \rrbracket)$, then it gets dominated at 
this step by Rule iii and stays that way afterwards.

\item \textbf{Every dominant position 
is isolated or has a private neighbour.}

Let us consider a dominant position $\textbf{u}$ which 
is not isolated. If it lies in $\llbracket 2,n-1\rrbracket \times \llbracket 2,m-1\rrbracket$, then it has a private neighbour since the pattern 
on $\llbracket 1,n\rrbracket \times \llbracket 1,m\rrbracket$ is a minimal dominating set of $G_{n,m}$. Else, it lies in some $\mathbb{B}_{n,m,k}$ for some 
$k \ge 0$ and there are two cases:

\begin{itemize}
\item \textbf{$\textbf{u}$ is not a corner.} Its neighbour $\textbf{v} \in \mathbb{B}_{n,m,k+1}$ 
is white by the application 
of the algorithm. Also, since its neighbours
in $\mathbb{B}_{n,m,k}$ are thus dominant or 
dominated, their neighbours in $\mathbb{B}_{n,m,k+1}$ are white. Moreover, the neighbour of $\textbf{v}$ in $\mathbb{B}_{n,m,k+2}$ is thus white. 
This is illustrated in Figure~\ref{figure.proof.private.neighbour.minimal}.
As a consequence $\textbf{v}$ is a private neighbour for $\textbf{u}$.

\item \textbf{$\textbf{u}$ is a corner.} We apply a similar reasoning.

\end{itemize}

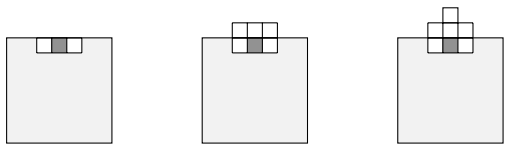
\begin{figure}[h!]
\[\begin{tikzpicture}[scale=0.2]
\fill[color=gray!10] (0,0) rectangle (7,7);
\fill[white] (2,6) rectangle (5,7);
\fill[gray!85] (3,6) rectangle (4,7);
\draw (2,6) grid (5,7);
\draw (0,0) rectangle (7,7);

\begin{scope}[xshift=13cm]
\fill[color=gray!10] (0,0) rectangle (7,7);
\fill[white] (2,6) rectangle (5,7);
\fill[gray!85] (3,6) rectangle (4,7);
\draw (2,6) grid (5,7);
\draw (0,0) rectangle (7,7);
\draw (2,7) grid (5,8);
\end{scope}

\begin{scope}[xshift=26cm]
\fill[color=gray!10] (0,0) rectangle (7,7);
\fill[white] (2,6) rectangle (5,7);
\fill[gray!85] (3,6) rectangle (4,7);
\draw (2,6) grid (5,7);
\draw (0,0) rectangle (7,7);
\draw (2,7) grid (5,8);
\draw (3,8) rectangle (4,9);
\end{scope}
\end{tikzpicture}\]
\caption{\label{figure.proof.private.neighbour.minimal} Illustration of the proof of a private neighbour 
for a non-isolated position. Steps of 
the completion algorithm for $X^M$ applied from left to right.}
\end{figure}

\end{itemize}
\end{enumerate}

\vspace{5cm}
\item \textbf{First inequality.} 

\begin{enumerate}
\item \textbf{Transformating patterns of $X^M$ into minimal dominating sets.} 
Let us define an application $\phi_{n,m}$ 
which, to each pattern of $X^M$ on $\llbracket 1,n\rrbracket \times \llbracket 1, m\rrbracket$, associates a minimal 
dominating set of $G_{n,m}$ defined by: 

\begin{enumerate}
\item suppressing 
any dominant position in $\mathbb{B}_{n,m,0}$ which has no private neighbours in $G_{n,m}$ and which is dominated 
by an element of $G_{n,m}$;
\item changing successively any non-dominant position of $\mathbb{B}_{n,m,0}$ which is still not 
dominated into
a dominant one;
\item successively, for every dominant position $\textbf{u} \in \mathbb{B}_{n,m,0}$ : if one of $\textbf{u}$'s neighbours $\textbf{v}$ is the only private neighbour of a position $\textbf{w}$ which is not isolated in $G_{n,m}$ then change $\textbf{w}$ into a non-dominant position.

This Step is illustrated 
on Figure~\ref{figure.third.step.minimal}.

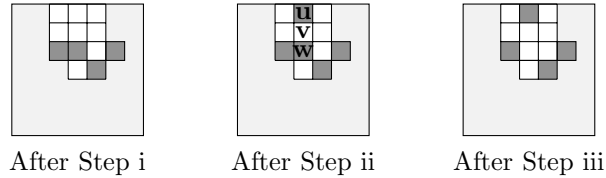
\begin{figure}[h!]
\[\begin{tikzpicture}[scale=0.25]
\fill[gray!10] (0,0) rectangle (7,7);
\fill[white] (2,4) rectangle (5,7);
\fill[gray!85] (2,4) rectangle (4,5);
\fill[white] (3,3) rectangle (4,4);
\fill[gray!85] (4,3) rectangle (5,4);
\fill[gray!85] (5,4) rectangle (6,5);
\draw (0,0) rectangle (7,7);
\draw (2,4) grid (5,7) ;
\draw (3,3) grid (5,4);
\draw (5,4) rectangle (6,5);
\node at (3.5,-1.5) {After Step i};

\begin{scope}[xshift=12cm]
\fill[gray!10] (0,0) rectangle (7,7);
\fill[white] (2,4) rectangle (5,7);
\fill[gray!85] (2,4) rectangle (4,5);
\fill[white] (3,3) rectangle (4,4);
\fill[gray!85] (3,6) rectangle (4,7);
\fill[gray!85] (4,3) rectangle (5,4);
\fill[gray!85] (5,4) rectangle (6,5);
\draw (0,0) rectangle (7,7);
\draw (2,4) grid (5,7) ;
\draw (3,3) grid (5,4);
\draw (5,4) rectangle (6,5);
\node at (3.5,-1.5) {After Step ii};
\node at (3.5,5.5) {\textbf{v}};
\node at (3.5,4.5) {\textbf{w}};
\node at (3.5,6.5) {\textbf{u}};
\end{scope}

\begin{scope}[xshift=24cm]
\fill[gray!10] (0,0) rectangle (7,7);
\fill[white] (2,4) rectangle (5,7);
\fill[gray!85] (2,4) rectangle (3,5);
\fill[white] (3,3) rectangle (4,4);
\fill[gray!85] (3,6) rectangle (4,7);
\fill[gray!85] (4,3) rectangle (5,4);
\fill[gray!85] (5,4) rectangle (6,5);
\draw (0,0) rectangle (7,7);
\draw (2,4) grid (5,7) ;
\draw (3,3) grid (5,4);
\draw (5,4) rectangle (6,5);
\node at (3.5,-1.5) {After Step iii};
\end{scope}
\end{tikzpicture}\]
\caption{\label{figure.third.step.minimal} Illustration of the second and then third steps 
of the algorithm defining $\phi_{n,m}$ for $X^M$, from 
left to right. $\textbf{u},\textbf{v}$ and $\textbf{w}$ are instances of the positions described in Rule iii.}
\end{figure}
\end{enumerate}

\item \textbf{Verifying that images of $\phi_{n,m}$ 
are minimal dominating sets.} 

Let us consider a globally admissible pattern $p$ of $X^M$ on $\llbracket 1, n\rrbracket \times \llbracket 1,m\rrbracket$. The set $\phi_{n,m} (p)$ is 
a minimal dominating set of $G_{n,m}$:

\begin{itemize}
\item \textbf{Any vertex of $G_{n,m}$ is dominated
or dominant in $\phi_{n,m} (p)$.}

Before Step ii, if a position is not dominant 
and not dominated, it becomes dominant during 
this step. 
Moreover, during Step iii, any position 
which is modified is transformed into a dominated position. 

\begin{comment}
A non-dominant position in $\phi_{n,m} (p)$ is either obtained by suppression of a dominant position in the first step (in which case it is dominated in $\phi_{n,m} (p)$), or was already dominated in $p$. In this case we have multiple possibilities: 

\begin{itemize}
\item 

If a vertex is a private neighbour in 
a $p$, then 
its dominant position is not suppressed in 
the first step of the algorithm. It may be suppressed in the third step, but it is replaced 
by another dominant position.
\item 
If it is not a private neighbour, but dominated, 
it stays also dominated after this step. 
\end{itemize}

\end{comment}

\item \textbf{Any non-isolated dominant position 
has a private neighbour.}
After applying $\phi_{n,m}$, only the positions on 
the border $\mathbb{B}_{n,m,0}$ may not have any private neighbour. After Step i, every

 dominant position on $\mathbb{B}_{n,m,0}$ is isolated, or has a private neighbour. After Step ii, some positions may be dominant, non-isolated, and have no private neighbours. Such positions become non-dominant in Step iii.

\end{itemize}

\item For all $n,m$, the number of preimages 
of $\phi_{n,m}$ for any minimal dominating set 
of $G_{n,m}$ is bounded (roughly) by $2^{3(2n+2m)}$, 
since any symbol modified by the application 
is at distance at most 2 from $\mathbb{B}_{n,m,0}$.
As a consequence, $N_{n,m} (X^M) \le 2^{6(n+m)} M_{n,m}$.

\end{enumerate}
\end{enumerate}
\end{proof}

\begin{lemma}
\label{lemma.comparison.minimal.total}
For all $n$, the following bounds hold: 
\[\frac{1}{2^{8(m+n)}} N_{n,m} (X^{MT}) \le MT_{n,m} \le N_{n,m} (X^{MT}).\]
\end{lemma}

For readability, we reproduce the structure 
of the proof of Lemma~\ref{lemma.comparison.minimal}, 
but simplify the arguments and refer this proof.

\begin{proof}
\leavevmode
\begin{enumerate}
\item \textbf{Second inequality.} 

\begin{enumerate}
\item \textbf{A completion algorithm 
of dominating set into a configuration of 
$X^{MT}$.} 

Consider a minimal total dominating set of $G_{n,m}$. Any element in $\llbracket 2, n-1 \rrbracket \times \llbracket 2, m-1 \rrbracket$ 
is dominated 
by an element of $\llbracket 1, n \rrbracket \times \llbracket 1, m \rrbracket$, and any dominant 
element in $\llbracket 2, n-1 \rrbracket \times \llbracket 2, m-1 \rrbracket$ is not isolated and has a private 
neighbour in $\llbracket 1, n \rrbracket \times \llbracket 1, m \rrbracket$ (which may or may not 
be a dominant position). Let us extend it 
into a configuration $x$ of $X^{MT}$ using 
an algorithm very similar to the one in the corresponding 
point in the proof of Lemma~\ref{lemma.comparison.minimal}, but the condition in the third point is different:
\begin{enumerate}
	\item if $\textbf{u}$ is a corner then $x_{\textbf{u}}$ is white;
	\item if $\textbf{u}$ is a neighbour of a corner in one of the vertical sides of $\mathbb{B}_{n,m,k+1}$ then $x_{\textbf{u}}$ is white;
	\item for all other $\textbf{u}$, $x_{\textbf{u}}$ is grey if and only if its neighbour in $\mathcal{N}^k (\llbracket 1,n \rrbracket \times \llbracket 1,m \rrbracket) $ is not dominated by an element in this set.
\end{enumerate}

\item \textbf{The result of the algorithm 
is a configuration of $X^{MT}$.}

\begin{itemize}
\item \textbf{Every position is dominated.} 
Similar to the corresponding point in 
the proof of Lemma~\ref{lemma.comparison.minimal}. This implies that no dominant positions are isolated.
\item \textbf{Every dominant position 
has a private neighbour.}

Let us consider a dominant position $\textbf{u}$. 
If it is in $\llbracket 3,n-2\rrbracket \times \llbracket 3,m-2\rrbracket $, 
since the pattern on $\llbracket 1,n\rrbracket \times \llbracket 1,m\rrbracket$ is a minimal total dominating set of $G_{n,m}$, we know that it has a private neighbour. Else, it lies
in some $\mathbb{B}_{n,m,k}$ for $k \ge 0$, 
or in $\llbracket 2,n-1\rrbracket \times \llbracket 2,m-1\rrbracket$. Then there are two cases: 

\begin{itemize}
\item \textbf{$\textbf{u}$ is not a corner.} If it has no dominant neighbours in 
Let us call $\textbf{v}$ its neighbour in $\mathbb{B}_{n,m,k+1}$. Note that, depending on whether or not $\textbf{u}$ is dominated inside $\mathcal{N}^k (\mathbb{B}_{n,m,0})$, \textbf{v} may be white or grey. Since the neighbours of $\textbf{u}$ in $\mathbb{B}_{n,m,k}$ are dominated, $\textbf{v}$'s neighbours in $\mathbb{B}_{n,m,k+1}$ are white. Finally, since $\textbf{v}$ is dominated by $\textbf{u}$, its neighbour in $\mathbb{B}_{n,m,k+2}$ is white, hence $\textbf{v}$ is a private neighbour for $\textbf{u}$.

\item \textbf{ $\textbf{u}$ is a corner.} We apply a similar reasoning.

\end{itemize}

\end{itemize}
\end{enumerate}

\item \textbf{First inequality:} 

\begin{enumerate}
\item \textbf{A transformation of patterns of $X^{MT}$ into dominating sets.} 

Let us define once again an application $\phi_{n,m}$ 
which, to each pattern of $X^{MT}$ on $\llbracket 1, n\rrbracket \times \llbracket 1, m\rrbracket$, associates a minimal total 
dominating set of $G_{n,m}$, defined in a similar
way as in the corresponding point 
in the proof of Lemma~\ref{lemma.comparison.minimal}, but the proof is more complex.

\begin{enumerate}
\item suppress 
any dominant position on the border 
$\mathbb{B}_{n,m,0}$ which has no private neighbours in $G_{n,m}$.
\item Successively, for every non-corner undominated position $\textbf{u}$ on 
the border $\mathbb{B}_{n,m,0}$, do the following:

\begin{itemize}
\item Consider the position $\textbf{v}$, neighbour of $\textbf{u}$ in $\llbracket 2,n-1\rrbracket \times \llbracket 2,m-1 \rrbracket$. For each dominant 
position $\textbf{w}$ in the neighbourhood of $\textbf{v}$, 
and for each dominant position $\textbf{w'}$ in the 
neighbourhood of $\textbf{w}$, if $\textbf{w}$ 
is the only private neighbour of $\textbf{w'}$, 
then change $\textbf{w'}$ into a non-dominant 
position.
\item Change $\textbf{v}$ into a dominant 
position.
\end{itemize}

Then do the same operations for the corners 
of $\mathbb{B}_{n,m,0}$, except that 
$\textbf{v}$ is replaced by any neighbour of the corner.

\end{enumerate}

This Step is illustrated 
on Figure~\ref{figure.third.step.minimal.total}.

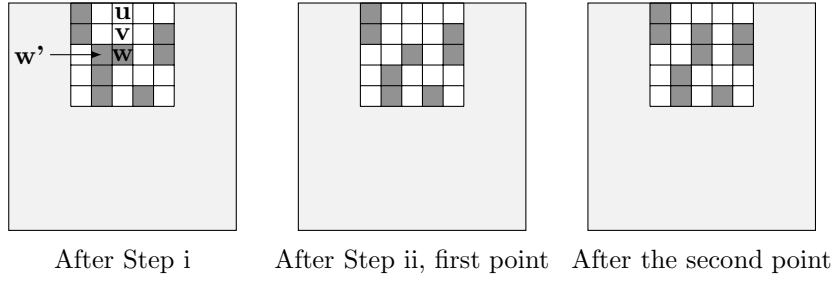
\begin{figure}[H]
\[\begin{tikzpicture}[scale=0.275]
\fill[gray!10] (-2,-2) rectangle (9,9);
\fill[white] (1,4) rectangle (6,9);
\fill[gray!85] (2,6) rectangle (4,7);
\fill[gray!85] (5,6) rectangle (6,7);
\fill[gray!85] (1,7) rectangle (2,9);
\fill[gray!85] (2,4) rectangle (3,6);
\fill[gray!85] (4,4) rectangle (5,5);
\fill[gray!85] (5,6) rectangle (6,8);
\draw (-2,-2) rectangle (9,9);
\draw (1,4) grid (6,9) ;
\node at (3.5,6.5) {$\textbf{w}$};
\node at (-1,6.5) {$\textbf{w'}$};
\draw[-latex] (0,6.5) -- (2.5,6.5);
\node at (3.5,7.5) {$\textbf{v}$};
\node at (3.5,8.5) {$\textbf{u}$};
\node at (3.5,-3.5) {After Step i};

\begin{scope}[xshift=14cm]

\fill[gray!10] (-2,-2) rectangle (9,9);
\fill[white] (1,4) rectangle (6,9);
\fill[gray!85] (3,6) rectangle (4,7);
\fill[gray!85] (5,6) rectangle (6,7);

\fill[gray!85] (1,7) rectangle (2,9);
\fill[gray!85] (2,4) rectangle (3,6);
\fill[gray!85] (4,4) rectangle (5,5);
\fill[gray!85] (5,6) rectangle (6,8);
\draw (-2,-2) rectangle (9,9);
\draw (1,4) grid (6,9) ;
\node at (3.5,-3.5) {After Step ii, first point};
\end{scope}

\begin{scope}[xshift=28cm]
\fill[gray!10] (-2,-2) rectangle (9,9);
\fill[white] (1,4) rectangle (6,9);
\fill[gray!85] (3,6) rectangle (4,7);
\fill[gray!85] (5,6) rectangle (6,7);
\fill[gray!85] (1,7) rectangle (2,9);
\fill[gray!85] (2,4) rectangle (3,6);
\fill[gray!85] (4,4) rectangle (5,5);
\fill[gray!85] (5,6) rectangle (6,8);

\fill[gray!85] (3,7) rectangle (4,8);
\draw (-2,-2) rectangle (9,9);
\draw (1,4) grid (6,9) ;
\node at (3.5,-3.5) {After the second point};
\end{scope}
\end{tikzpicture}\]

\caption{Illustration of the second and then third steps 
of the algorithm defining $\phi_{n,m}$ for $X^MT$, from 
left to right. $\textbf{u},\textbf{v},\textbf{w}$ and $\textbf{w'}$ are instances 
of the positions described in Rule ii.}
\label{figure.third.step.minimal.total}
\end{figure}

\item \textbf{Verification that images of $\phi_{n,m}$ 
are minimal total dominating sets.} 

Consider a pattern $p$ of $X^{MT}$ on $\llbracket 1,n\rrbracket \times \llbracket 1, m\rrbracket$. The set $\phi_{n,m} (p)$ is 
a minimal total dominating set of $G_{n,m}$:

%TODO change text proof because no longer good explanation
\begin{itemize}
\item \textbf{Any vertex of $G_{n,m}$ is dominated
in $\phi_{n,m} (p)$.}

Any (dominant or not) position which was dominated before applying Rule i is still dominated afterwards: indeed, if some position $\textbf{u}$ lies in the neighbourhood of a dominant position $\textbf{v}$ suppressed by 
Rule i, then since $\textbf{v}$ had no private neighbours in $G_{n,m}$, $\textbf{u}$ is dominated by another position.
For similar reasons, no positions become undominated after the application of Rule ii: only the neighbours of some $\textbf{w'}$ could be affected and if \textbf{w'} becomes non-dominant it means that they were dominated by other positions, so that they stay dominated.
Since all the positions inside $\llbracket 2,n-1\rrbracket \times \llbracket 2,m-1 \rrbracket$ were dominated before applying the rules, it only remains to show that the positions inside $\mathbb{B}_{n,m,0}$ are dominated after applying Rule ii. This is true thanks to this rule: any undominated position $\textbf{u}$ inside the border sees its neighbour $\textbf{v}$ inside  $\llbracket 2,n-1\rrbracket \times \llbracket 2,m-1 \rrbracket$ becomes dominant. The same works for the corners, except that the neighbour comes from the border.

\item \textbf{Any dominant position 
has a private neighbour.}

At the end of Step i, any dominant position has a private neighbour. Only the creation of 
a domination position $\textbf{v}$ during the execution of Rule ii on position $\textbf{u}$ could affect this property by disabling the private neighbour of a position $\textbf{w}$ in its neighbourhood or by not having any private neighbour itself. The first case cannot happen since any dominant position $\textbf{w'}$ whose unique private neighbour is $\textbf{w}$ is suppressed. The second one also never happens since the position $\textbf{u}$ is a private neighbour for $\textbf{v}$. 
\end{itemize}

\item For all $n,m$, the number of preimages 
of $\phi_{n,m}$ for any minimal dominating set 
of $G_{n,m}$ is bounded (roughly) by $2^{4*(2m+2n)}$, 
since any symbol modified by the application 
is at distance at most $4$ of the border of $\llbracket 1,n\rrbracket \times \llbracket 1,m \rrbracket$. 
As a consequence, $N_{n,m} (X^{MT}) \le 2^{8(m+n)} T_{n,m}$.

\end{enumerate}
\end{enumerate}
\end{proof}

\begin{thm}[Asymptotic behaviour]
There exists some $\nu_D \geq 0$ (resp. 
$\nu_M$, $\nu_T$ and $\nu_{MT}$) 
such that

\[ \boxed{D_{n,m} = \nu_D^{ nm + o(nm)}}\] 
(resp. $M_{n,m} = \nu_M^{ nm + o(nm)}$, $T_{n,m} = \nu_T^{nm + o(nm)}$ and $MT_{n,m} = \nu_{MT}^{ nm + o(nm)}$).
\end{thm}

\begin{proof}
Let us prove this for the sequence $(M_{n,m})$ (the proof is similar for the other sequences).

As a consequence of 
Lemma~\ref{lemma.comparison.minimal}, for all $n,m$, 
\[- \frac{12(m+n)}{nm} + \frac{\log_2 (N_{n,m} (X^M))}{n m}
\le \frac{\log_2(M_{n,m})}{nm} \le \frac{\log_2 (N_{n,m} (X^M))}{nm}.\]

As a consequence, 
\[\frac{\log_2(M_{n,m})}{nm} \rightarrow h(X^M).\]
This means that $M_{n,m} = 2^{h(X^M) \cdot nm+ o(nm)} = \nu_M^{nm+o(nm)}$, 
where $\nu_M = 2^{h(X^M)}$. %TODO vérifier que vrai
%TODO vérifier ailleurs

\end{proof}

\section{\label{section.block.gluing.property}Computability of the growth rate}
In this section, we prove that the growth rate $\nu_D$ (resp. $\nu_M$, $\nu_T$ and $\nu_{MT}$) is a computable number, meaning that there 
exists an algorithm which computes 
approximations of this number with 
arbitrary given precision. For this purpose, we rely on the block-gluing property, 
defined in Section~\ref{subsection.block.gluing}, and proved for $X^D$ (resp. $X^M$, $X^T$ and $X^{MT}$) in Section~\ref{subsection.proof.block.gluing}. 
If a subshift of finite has this property then its entropy is computable. We describe a known algorithm to compute it.

\subsection{\label{subsection.block.gluing} 
The block-gluing property}

\subsubsection{Definition}

For two finite subsets $\mathbb{U},\mathbb{V}$ 
of $\mathbb{Z}^2$, we write
\[\delta(\mathbb{U},\mathbb{V})=\min_{\textbf{u} \in \mathbb{U}} \min_{\textbf{v} \in \mathbb{V}} ||\textbf{v}-\textbf{u}||_{\infty}.\]
The usual definition of the block-gluing property is the following one:

\begin{definition}
For a fixed integer $c \geq 0$, we say that a bidimensional subshift of finite type $X$ on alphabet $\mathcal{A}$ is $c$-block-gluing when for every $n \geq 0$ and any two globally admissible patterns $p$ and $q$ of $X$ on support 
$\llbracket 1,n \rrbracket ^2$, for all $\textbf{u},\textbf{v} \in \mathbb{Z}^2$ such that $\delta(\textbf{u}+\llbracket 1,n \rrbracket ^2, \textbf{v}+\llbracket 1,n \rrbracket ^2) \ge c$, 
there exists a configuration $x$ of $X$ 
such that $x_{\textbf{u}+\llbracket 1,n \rrbracket ^2} = p$ 
and $x_{\textbf{v}+\llbracket 1,n \rrbracket ^2} = q$. 
\end{definition}

Informally, this means that any pair of rectangular patterns placed at whatever positions can be completed 
into a configuration of $X$, provided that the 
distance between the two patterns is at least $c$.

\begin{notation}
For any subshift of finite type $X$, 
we denote by $c(X)$ the smallest $c$ such that 
$X$ is $c$-block-gluing. If $X$ is not block 
gluing for any integer $c$, we write $c(X)=+\infty$.
\end{notation}

In the following, we will use the notations $\mathbb{Z}_{-} =\, \rrbracket - \infty,0 \rrbracket$ and $\mathbb{Z}_{+} = \,\rrbracket 0,+ \infty \llbracket$.
Here is a characterisation of the block-gluing property:

\begin{proposition}
\label{proposition.semi.plane}
Let $c \ge 0$ be an integer.
A bidimensional subshift $X$ is $c$-block-gluing
if and only if for all $k \ge c$ and 
$p$ and $q$ globally admissible patterns on 
supports $\mathbb{Z}_{-} \times \mathbb{Z}$ (resp. $\mathbb{Z} \times \mathbb{Z}_{-}$) and 
$\mathbb{Z}_{+} \times \mathbb{Z}$ (resp. $\mathbb{Z} \times \mathbb{Z}_{+}$), there 
exists a configuration $x$ in $X$ such that 
$x_{|\mathbb{Z}_{-} \times \mathbb{Z}} = p$ 
and $x_{|(k,0)+\mathbb{Z}_{+} \times \mathbb{Z}}=q$ (resp. $x_{|\mathbb{Z} \times \mathbb{Z}_{-}} = p$ 
and $x_{|(k,0)+\mathbb{Z} \times \mathbb{Z}_{+}}=q$).
\end{proposition}

Informally, in order to check the block-gluing 
property, it is sufficient to prove 
that any two patterns on half-planes can be 
glued with arbitrary distance greater than $c$ 
in a configuration of $X$.

\begin{proof} \leavevmode
\begin{itemize}
\item $(\Leftarrow)$: Let us assume that $X$ verifies the second hypothesis. Let us consider 
some integer $n$, and two globally admissible
patterns $\overline{p},\overline{q}$ of $X$ on support $\llbracket 1,n \rrbracket^2$. 
Let $\textbf{u},\textbf{v}$ be two positions such that 
$\delta(\textbf{u}+\llbracket 1,n \rrbracket^2, \textbf{v}+\llbracket 1,n \rrbracket^2)\ge c$. This means that 
the two translates $\textbf{u}+\llbracket 1,n \rrbracket^2$ 
and $\textbf{v}+\llbracket 1,n \rrbracket^2$
have more than $c$ columns separating 
them or more than $c$ rows. Without loss 
of generality, we assume that we are in 
the case of separating columns, and denote by 
$k \ge c$ the exact number of columns separating $\overline{p}$ and $\overline{q}$. Since $\overline{p}$ and $\overline{q}$ are globally 
admissible, there exist $p$ and $q$ globally admissible 
patterns of $X$ on respective supports 
$\mathbb{Z}_{-} \times \mathbb{Z}$ 
and $(k,0)+ \mathbb{Z}_{+} \times \mathbb{Z}$
whose restrictions on 
$ -(n,0) + \llbracket 1,n \rrbracket^2$ 
and $\textbf{v} -\textbf{u} - (n,0) + \llbracket 1,n \rrbracket^2$ are respectively $\overline{p}$ 
and $\overline{q}$. 
By hypothesis, there exists some configuration 
$x$ of $X$ whose restrictions on $\mathbb{Z}_{-} \times \mathbb{Z}$ 
and $(k,0)+ \mathbb{Z}_{+} \times \mathbb{Z}$
are respectively $p$ and $q$. The patterns 
$\overline{p}$ and $\overline{q}$ can be found 
on $\textbf{u} + \llbracket 1, n \rrbracket ^2$ 
and $\textbf{v} + \llbracket 1, n \rrbracket ^2$
in the configuration $\sigma^{\textbf{u}} x$. 

\item $(\Rightarrow)$: Let us assume the first hypothesis on $X$ is true, and let $p$ and $q$ be two patterns on supports $\mathbb{Z}_{-} \times \mathbb{Z}$ and $(k,0)+\mathbb{Z}_{+} \times \mathbb{Z}$ 
for some $k \ge c$ (the other case is proved in a similar way). From the block-gluing 
property, for all $n \ge 0$
one can extend the restriction of $p$ on 
$\llbracket 1, n \rrbracket^2 - (n,0)$ and the restriction 
of $q$ on $\llbracket 1, n \rrbracket^2 - (n,0)+(k+1,0)$
into a configuration $x_n \in X$. 
By compactness of the set $X$ for 
the product of the discrete topology, this sequence admits 
a subsequence which converges to some $x \in X$. This $x$ verifies 
$x_{|\mathbb{Z}_{-} \times \mathbb{Z}}=p$ 
and $x_{|(k,0)+\mathbb{Z}_{+} \times \mathbb{Z}}=q$.
\end{itemize}
\end{proof}

\subsubsection{\label{subsection.computability} Algorithmic computability of 
entropy}

\begin{definition}
Let $f: \mathbb{N} \rightarrow \mathbb{N}$ 
a computable function.
A real number $x$ is said to be \textbf{computable} with rate $f$ when there exists an algorithm which, given an integer $n$ as input, outputs in at most $f(n)$ steps a rational number $r_n$ 
such that $|x-r_n|\le \frac{1}{n}$.
\end{definition}

This definition corresponds to Definition 1.3 in~\cite{Pavlov-Schraudner}. The following theorem is Theorem 1.4 in the same reference. Its proof provides an algorithm to compute $h(X)$.

\begin{thm}[\cite{Pavlov-Schraudner}]
\label{theorem.computability.entropy}
Let $X$ be a block-gluing bidimensional subshift of finite type. 
Then $h(X)$ is computable with rate $n \mapsto 2^{O(n^2)}$.
\end{thm}

\begin{lemma}
\label{lemma.counting.block.gluing}
Let $X$ be a $c$-block gluing 
bidimensional subshift of finite type 
on alphabet $\mathcal{A}$. For all $k \ge 1$, 
the number $N_k (X)$ is equal to the number of $k \times k$
 patterns which appear in a $\left(|\mathcal{A}|^{2c+1} \cdot (c+k)+1 \right) \times (2c+k+2)$ locally-admissible rectangular pattern whose restrictions on the two extremal vertical (resp. 
 horizontal) sides are equal.
\end{lemma}

\begin{remark}
Let us note that in general the entropy 
of a bidimensional subshift of finite type is not computable at all (\cite{Hochman-Meyerovitch} Theorem 1.1 and the existence 
of non-computable right recursively enumerable 
numbers).
\end{remark}

This algorithm is as follows:
\begin{center} \begin{algorithm}[H]
\SetAlgoLined
\SetKwData{Left}{left}
\SetKwData{This}{this}
\SetKwData{Up}{up}
\SetKwFunction{Union}{Union}
\SetKwFunction{FindCompress}{FindCompress}

\SetKwInput{Input}{Input}
\SetKwInOut{Output}{Output}
\Input{An integer $n$, an alphabet $\mathcal{A}$ and a set of patterns $\mathcal{F}$ of $\mathcal{A}^{\mathbb{U}}$ for 
some finite $\mathbb{U} \subset \mathbb{Z}^2$}
\Output{A rational approximation of $h(X)$ up to $1/n$, where $X$ is the SFT on alphabet $\mathcal{A}$ defined by the set of forbidden patterns $\mathcal{F}$}
 $k \leftarrow 0$\\
 $r \leftarrow +\infty$\\
 \While {$r \geq 1/2n$}
 {
 $k \leftarrow k + 1$\\
 $m \leftarrow N_k (X)$ (this is a sub-procedure 
 using Lemma~\ref{lemma.counting.block.gluing}).
 
 $r \leftarrow $ some rational approximation up to $1/2k$ of $\frac{\log_2 (N_k(X))}{k^2} - \frac{\log_2 (N_k(X))}{(k+c)^2}$
 } 
 
 Return a rational approximation up to $1/2n$ of $\log_2 (N_k (X))/k^2$ 
 \caption{Computing the entropy of a $c$-block-gluing bidimensional SFT.}
\end{algorithm}
\end{center}

\subsection{\label{subsection.proof.block.gluing}Proof of the 
block-gluing property for domination subshifts}

It is straightforward to check that 
the domination subshift $X^D$ and the total domination subshift $X^T$
satisfy the block-gluing property, 
with $c(X^D)=1$ (just fill every cell with grey). In 
this section, we prove that 
$X^M$ and $X^{MT}$ also satisfy this property.

\begin{notation}
In the following, for all $j \in \mathbb{Z}$, 
we denote by $C_j$ the column $\{j\} \times \mathbb{Z}$ of $\mathbb{Z}^2$.
\end{notation}

\begin{thm}
\label{minimal-gluing-th}
The minimal domination subshift is block gluing and $c(X^M)=5$.
\end{thm}

\noindent \textbf{Idea of the proof:}
\textit{In order to simplify the proof of the block-gluing property, we rely on 
Proposition~\ref{proposition.semi.plane}.
The proof of the block-gluing property for 
two half-plane patterns consists in 
determining successively the intermediate columns from the patterns towards
the "center" (chosen to be column $C_{k-2}$, for concision). The completion follows an algorithm which ensures 
that, when the number of intermediate columns 
is great enough, any added dominant element 
has a private neighbour in an already 
constructed column or is isolated. This ensures that 
the rules of the subshift are not broken.}

\begin{proof} \leavevmode

\begin{figure}[h!]
\[\begin{tikzpicture}[scale=0.4]
\draw (0,0) grid (1,3);
\draw (-1,1) rectangle (0,2);
\node at (0.5,-0.5) {$\vdots$};
\node at (0.5,4) {$\vdots$};

\draw[-latex] (2,1.5) -- (4,1.5);

\begin{scope}[xshift=6cm]
\node at (0.5,-0.5) {$\vdots$};
\node at (0.5,4) {$\vdots$};
\fill[gray!85] (1,1) rectangle (2,2);
\draw (0,0) grid (1,3);
\draw (-1,1) grid (2,2);
\end{scope}
\end{tikzpicture}\]
\caption{\label{figure.intermediate.columns} Illustration of the rule for filling the non-central intermediate columns for $X^{M}$.}
\end{figure}
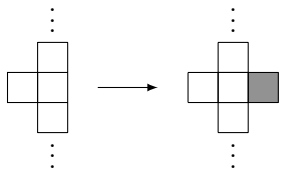

\begin{figure}[h!]
\centering
\begin{tikzpicture}[scale=0.3]
%-------------First one

\fill[gray!85] (0,3) rectangle (1,8);
\fill[gray!85] (0,0) rectangle (1,1);
\fill[gray!85] (1,2) rectangle (2,3);

\fill[gray!85] (10,3) rectangle (11,8);
\fill[gray!85] (8,1) rectangle (10,2);

\foreach \x in {-1,...,2} {\draw[mydashed] (\x, -1) -- (\x, 0);}
\foreach \x in {-1,...,2} {\draw[mydashed] (\x, 8) -- (\x, 9);}
\foreach \x in {8,...,11} {\draw[mydashed] (\x, -1) -- (\x, 0);}
\foreach \x in {8,...,11} {\draw[mydashed] (\x, 8) -- (\x, 9);}

\foreach \y in {0,...,8} {\draw[mydashed] (-1,\y) -- (-2,\y);}
\foreach \y in {0,...,8} {\draw[mydashed] (11,\y) -- (12,\y);}

\draw (-1,0) grid (2,8);
\draw (8,0) grid (11,8);
\node at (-3,4) {$p$};
\node at (13,4) {$q$};
\node[scale=0.625] at (2.6,-2) {$C_1$};
\node[scale=0.625] at (1.6,-2) {$C_0$};
\node[scale=0.625] at (8.95,-2) {$C_{k+1}$};
\node[scale=0.625] at (7.7,-2) {$C_k$};
\node[scale=0.7] at (5,-2) {$\hdots$};

\draw[line width =0.5mm] (2,-1.5) -- (2,9.5);
\draw[line width =0.5mm] (8,-1.5) -- (8,9.5);

\node[scale=1.25] at (5,-3.5) {$(0)$};

%------------------Second one
\begin{scope}[xshift=20cm]

\fill[gray!85] (0,3) rectangle (1,8);
\fill[gray!85] (0,0) rectangle (1,1);
\fill[gray!85] (1,2) rectangle (2,3);
\fill[gray!85] (10,3) rectangle (11,8);
\fill[gray!85] (8,1) rectangle (10,2);

\fill[gray!85] (3,3) rectangle (4,7);
\fill[gray!85] (3,1) rectangle (4,2);
\fill[gray!85] (7,3) rectangle (8,7);

\foreach \x in {-1,...,2} {\draw[mydashed] (\x, -1) -- (\x, 0);}
\foreach \x in {-1,...,2} {\draw[mydashed] (\x, 8) -- (\x, 9);}
\foreach \x in {8,...,11} {\draw[mydashed] (\x, -1) -- (\x, 0);}
\foreach \x in {8,...,11} {\draw[mydashed] (\x, 8) -- (\x, 9);}

\foreach \x in {3,4,5,6,7} {\draw[mydashed] (\x, -1) -- (\x, 0);}
\foreach \x in {3,4,5,6,7} {\draw[mydashed] (\x, 8) -- (\x, 9);}

\foreach \y in {0,...,8} {\draw[mydashed] (-1,\y) -- (-2,\y);}
\foreach \y in {0,...,8} {\draw[mydashed] (11,\y) -- (12,\y);}
\draw (-1,0) grid (2,8); 
\draw (8,0) grid (11,8);
\draw (2,0) grid (5,8);
\draw (6,0) grid (8,8);

\draw[white, line width=0.5mm] (3.03,0) -- (3.96,0); 
\draw[white, line width=0.5mm] (3.03,8) -- (3.96,8); 
\draw[white, line width=0.5mm] (5.03,8) -- (5.96,8); 
\draw[white, line width=0.5mm] (5.03,0) -- (6.96,0); 
\draw[white, line width=0.5mm] (6,-1) -- (6,0.95);

\draw[mydashed] (6, -1) -- (6, 1);

\draw[line width =0.5mm] (2,-1.5) -- (2,9.5);
\draw[line width =0.5mm] (8,-1.5) -- (8,9.5);

\node[scale=0.7] at (6,-2) {$C_{k-2}$};
\node[scale=1.25] at (5,-3.5) {$(1)$};

\end{scope}

%-------------Third one
\begin{scope}[yshift=-15cm]

\fill[gray!85] (0,3) rectangle (1,8);
\fill[gray!85] (0,0) rectangle (1,1);
\fill[gray!85] (1,2) rectangle (2,3);
\fill[gray!85] (10,3) rectangle (11,8);
\fill[gray!85] (8,1) rectangle (10,2);

\fill[gray!85] (3,3) rectangle (4,7);
\fill[gray!85] (3,1) rectangle (4,2);
\fill[gray!85] (7,3) rectangle (8,7);

\fill[gray!85] (5,1) rectangle (6,3);

\foreach \x in {-1,...,2} {\draw[mydashed] (\x, -1) -- (\x, 0);}
\foreach \x in {-1,...,2} {\draw[mydashed] (\x, 8) -- (\x, 9);}
\foreach \x in {8,...,11} {\draw[mydashed] (\x, -1) -- (\x, 0);}
\foreach \x in {8,...,11} {\draw[mydashed] (\x, 8) -- (\x, 9);}

\foreach \x in {3,4,5,6,7} {\draw[mydashed] (\x, -1) -- (\x, 0);}
\foreach \x in {3,4,5,6,7} {\draw[mydashed] (\x, 8) -- (\x, 9);}

\foreach \y in {0,...,8} {\draw[mydashed] (-1,\y) -- (-2,\y);}
\foreach \y in {0,...,8} {\draw[mydashed] (11,\y) -- (12,\y);}
\draw (-1,0) grid (2,8); 
\draw (8,0) grid (11,8);
\draw (2,0) grid (8,8);

\draw[white, line width=0.5mm] (3.03,0) -- (3.96,0); 
\draw[white, line width=0.5mm] (3.03,8) -- (3.96,8); 
\draw[white, line width=0.5mm] (5.03,8) -- (5.96,8); 
\draw[white, line width=0.5mm] (5.03,0) -- (6.96,0); 
\draw[white, line width=0.5mm] (6,-1) -- (6,0.95);

\draw[mydashed] (6,-1) -- (6,1);
\draw[line width =0.5mm] (2,-1.5) -- (2,9.5);
\draw[line width =0.5mm] (8,-1.5) -- (8,9.5);

\node[scale=1.25] at (5,-3) {$(2)$};
\end{scope}

%-------------Fourth one
\begin{scope}[yshift=-15cm,xshift=20cm]

\fill[gray!85] (0,3) rectangle (1,8);
\fill[gray!85] (0,0) rectangle (1,1);
\fill[gray!85] (1,2) rectangle (2,3);
\fill[gray!85] (10,3) rectangle (11,8);
\fill[gray!85] (8,1) rectangle (10,2);

\fill[gray!85] (3,3) rectangle (4,7);
\fill[gray!85] (3,1) rectangle (4,2);

\fill[gray!85] (7,3) rectangle (8,7);

%Central column below
\fill[gray!85] (5,1) rectangle (6,3);
\fill[gray!85] (5,4) rectangle (6,5);
\fill[gray!85] (5,6) rectangle (6,7);

\foreach \x in {-1,...,2} {\draw[mydashed] (\x, -1) -- (\x, 0);}
\foreach \x in {-1,...,2} {\draw[mydashed] (\x, 8) -- (\x, 9);}
\foreach \x in {8,...,11} {\draw[mydashed] (\x, -1) -- (\x, 0);}
\foreach \x in {8,...,11} {\draw[mydashed] (\x, 8) -- (\x, 9);}

\foreach \x in {3,4,5,6,7} {\draw[mydashed] (\x, -1) -- (\x, 0);}
\foreach \x in {3,4,5,6,7} {\draw[mydashed] (\x, 8) -- (\x, 9);}

\foreach \y in {0,...,8} {\draw[mydashed] (-1,\y) -- (-2,\y);}
\foreach \y in {0,...,8} {\draw[mydashed] (11,\y) -- (12,\y);}
\draw (-1,0) grid (2,8); 
\draw (8,0) grid (11,8);
\draw (2,0) grid (8,8);

\draw[white, line width=0.5mm] (3.03,0) -- (3.96,0); 
\draw[white, line width=0.5mm] (3.03,8) -- (3.96,8); 
\draw[white, line width=0.5mm] (5.03,8) -- (5.96,8); 
\draw[white, line width=0.5mm] (5.03,0) -- (6.96,0); 
\draw[white, line width=0.5mm] (6,-1) -- (6,0.95);

\draw[mydashed] (6,-1) -- (6,1);

\draw[line width =0.5mm] (2,-1.5) -- (2,9.5);
\draw[line width =0.5mm] (8,-1.5) -- (8,9.5);

\node[scale=1.25] at (5,-3) {$(3)$};
\end{scope}
\end{tikzpicture}
\caption{\label{figure.completing.algorithm} Illustration of the algorithm
filling the intermediate columns between two half-plane patterns $p$ and $q$ for the minimal domination.\\
(0) Initial setting of the two patterns.\\
(i) After Step i of the algorithm\\
Some cells between $p$ and $q$ are not forced by these patterns: we left them 
non-filled. We chose $k=6$, still the proof works with $k=5$.}
\end{figure}
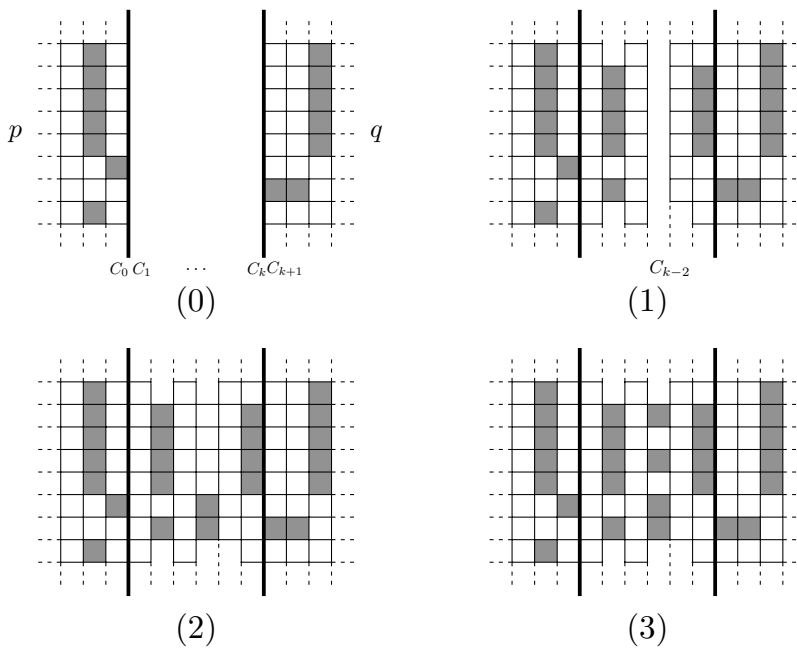

\begin{itemize}
\item \textbf{Filling the intermediate columns 
between two half-plane patterns.}

Let $p$ and $q$ be two patterns respectively on 
$\mathbb{Z}_{-} \times \mathbb{Z}$ and 
$\mathbb{Z}_{+} \times \mathbb{Z}$ (the 
proof for the other case is similar). Let 
us determine a configuration of $\mathcal{A}_0^{\mathbb{Z}^2}$ such that $x_{|\mathbb{Z}_{-} \times \mathbb{Z}}=p$ and $x_{|(k,0)+\mathbb{Z}_{+} \times \mathbb{Z}}=q$. 
The intermediate 
columns $C_1,...,C_k$ are determined by the following algorithm:

\begin{enumerate}
\vspace{-0,1cm}
\setlength{\parskip}{0pt}
\setlength{\itemsep}{3pt}
\item \textbf{Filling the intermediate columns, from $C_1$ to $C_{k-3}$, then $C_k,C_{k-1}$.} 

Successively, for all $j$ 
from $1$ to $k-3$, we determine the column 
$C_j$ according to the following rule: 
for all $\textbf{u} \in C_j$, $x_{\textbf{u}}$ 
is $\begin{tikzpicture}[scale=0.3,baseline=0.4mm]
\fill[gray!85] (0,0) rectangle (1,1);
\draw (0,0) rectangle (1,1);
\end{tikzpicture}$ when $x_{\textbf{u}-(1,0)}$, 
$x_{\textbf{u}-(1,1)}$, $x_{\textbf{u}-(1,-1)}$ 
and $x_{\textbf{u}-(2,0)}$ are $\begin{tikzpicture}
[scale=0.3,baseline=0.4mm]
\draw (0,0) rectangle (1,1);
\end{tikzpicture}$ (see Figure~\ref{figure.intermediate.columns}). Else, $x_{\textbf{u}}$ is set to $\begin{tikzpicture}[scale=0.3,baseline=0.4mm]
\draw (0,0) rectangle (1,1);
\end{tikzpicture}$.
Similarly, for $j=k$ and then $j=k-1$, 
we determine $x$ on any position $x_{\textbf{u}}$ 
for $\textbf{u} \in C_j$ by applying 
a symmetrical rule: $x_{\textbf{u}}$ 
is $\begin{tikzpicture}[scale=0.3,baseline=0.4mm]
\fill[gray!85] (0,0) rectangle (1,1);
\draw (0,0) rectangle (1,1);
\end{tikzpicture}$ when $x_{\textbf{u}+(1,0)}$, 
$x_{\textbf{u}+(1,1)}$, $x_{\textbf{u}+(1,-1)}$ 
and $x_{\textbf{u}+(2,0)}$ are
 $\begin{tikzpicture}[scale=0.3,baseline=0.4mm]
\draw (0,0) rectangle (1,1);
\end{tikzpicture}$. Else $x_{\textbf{u}}$ is set to
$\begin{tikzpicture}[scale=0.3,baseline=0.4mm]
\draw (0,0) rectangle (1,1);
\end{tikzpicture}$.

%\item Stage 2: Do the same for column 2 and 4: put an 'x' in a cell of column 2 (resp. 4) if its neighbour in column 1 (resp. 5) was not yet dominated;
%\item Stage 2': If $k > 5$, repeat Stage 2 for columns $3, \cdots, k-3$. For column $i$, the neighbour we consider is the one in column $i-1$;

%below: say first draft
\item \textbf{The central column $C_{k-2}$.} 

We now 
determine $x$ on the central column $C_{k-2}$. For all $\textbf{u} \in C_{k-2}$, $x_{\textbf{u}}$ 
is $\begin{tikzpicture}[scale=0.3,baseline=0.4mm]
\fill[gray!85] (0,0) rectangle (1,1);
\draw (0,0) rectangle (1,1);

\end{tikzpicture}$ when $x_{\textbf{u}+(1,0)}$, 
$x_{\textbf{u}+(1,1)}$, $x_{\textbf{u}+(1,-1)}$ 
and $x_{\textbf{u}+(2,0)}$ are equal to 
$\begin{tikzpicture}[scale=0.3,baseline=0.4mm]
\draw (0,0) rectangle (1,1);
\end{tikzpicture}$,
or when $x_{\textbf{u}-(1,0)}$, 
$x_{\textbf{u}-(1,1)}$, $x_{\textbf{u}-(1,-1)}$ 
and $x_{\textbf{u}-(2,0)}$ are equal to 
$\begin{tikzpicture}[scale=0.3,baseline=0.4mm]
\draw (0,0) rectangle (1,1);
\end{tikzpicture}$. Else it is 
$\begin{tikzpicture}[scale=0.3,baseline=0.4mm]
\draw (0,0) rectangle (1,1);
\end{tikzpicture}$.

\item \textbf{Eliminating non-domination errors 
in the central column.}
Choose any position $\textbf{u}_0 \in C_{k-2}$ and check if 
this position has a symbol $\begin{tikzpicture}[scale=0.3,baseline=0.4mm]
\fill[gray!85] (0,0) rectangle (1,1);
\draw (0,0) rectangle (1,1);
\end{tikzpicture}$ in its neighbourhood. 
If not, then set the symbol $\begin{tikzpicture}[scale=0.3,baseline=0.4mm]
\fill[gray!85] (0,0) rectangle (1,1);
\draw (0,0) rectangle (1,1);
\end{tikzpicture}$ on this position. Repeat this from $\textbf{u}_0+(0,1)$ upwards, and in parallel from $\textbf{u}_0-(0,1)$ downwards.

%\item Stage 3: For each cell in column $k-2$, put an 'x' if its neighbour in column $k-3$ or $k-1$ is not dominated;
%\item Stage 4: From cell $(k-2, 0)$ downwards, if a cell is not dominated, then put an 'x' inside. Do the same from $(k-2, 1)$ upwards.
\end{enumerate}

See an illustration of this algorithm on 
Figure~\ref{figure.completing.algorithm}.

\item \textbf{The obtained configuration 
is in $X^M$.}

We have to check that the configuration $x$ we constructed satisfies the local 
rules of the minimal domination subshift.

\begin{enumerate}
\item \textbf{Local rules are verified inside 
the half-planes.}

By hypothesis, the patterns $p$ and $q$ 
are globally admissible in $X^M$. As a consequence, for all $\textbf{u}$ 
in $\rrbracket - \infty,-2\rrbracket \times \mathbb{Z}$ or $\llbracket k+3,+\infty\llbracket \times \mathbb{Z}$, $x_{|\textbf{u}+\llbracket - 2,2 \rrbracket^2}$ is not a forbidden pattern.
We have left to check that no forbidden patterns are created through the execution 
of the algorithm described in the first point 
of the proof.

\item \textbf{Every position outside $\text{supp}(p) \cup \text{supp}(q)$ and not in $S$ is dominated.}

In the columns $C_{-1}$ and $C_{k+2}$, this 
comes from the fact that the patterns $p$ and $q$ 
are globally admissible. For $j$ between $0$ and $k-3$, 
and $\textbf{u} \in C_j$, if 
$\textbf{u}$ is not dominated by a position 
in $C_j$ or $C_{j-1}$
the position $\textbf{u}+(1,0)$ is the 
symbol $\begin{tikzpicture}[scale=0.3,baseline=0.4mm]
\fill[gray!85] (0,0) rectangle (1,1);
\draw (0,0) rectangle (1,1);
\end{tikzpicture}$ (by the first and second steps of the algorithm), and thus $\textbf{u}$ 
is dominated. A symmetrical reasoning 
works for the positions in the columns $C_{k+1},C_{k},C_{k-1}$. For a position 
in the central column $C_{k-2}$, this is guaranteed by Step 3.

\item \textbf{Every dominant position outside $\text{supp}(p) \cup \text{supp}(q)$ is isolated or has a private neighbour not in $S$.}

Let us consider a non-isolated dominant position $\textbf{u}$.
\begin{enumerate}
\item If it lies in $C_{-1}$ (resp. $C_{k+2}$), 
$\textbf{u}$ has a private neighbour
in a configuration of $X$ 
that extends $p$ (resp. $q$). 
If this private neighbour is in column $C_{-2}$ 
or $C_{-1}$ (resp. $C_{k+2}$ or $C_{k+3}$), then 
it stays a private neighbour of $\textbf{u}$ in $x$. If it is $\textbf{u} + (1,0)$ (resp. $\textbf{u} - (1,0)$), then it stays 
a private neighbour in $x$: since this position 
is dominated by $\textbf{u}$ according to the first step 
of the algorithm (resp. second step), it 
is not dominated in $x$ by a position in $C_0$ 
(resp. $C_{k+1}$). The same reasoning is applied
to positions in columns $C_0$ and $C_{k+1}$.
\item In the other columns $C_j$ for $j < k-2$,
the first step guarantees, for any position $\textbf{u}$ 
in $C_j$ that is a dominant position, that the position 
$\textbf{u}-(1,0)$ is a private neighbour. 
A similar reasoning applies to column $C_{k}$ 
and $C_{k-1}$.
\item If $\textbf{u}$ is in 
the column $C_{k-2}$, it means that it 
was introduced in either the second or the third step, 
meaning that it has a private neighbour in 
column $C_{k-3}$ or $C_{k-1}$. 
\end{enumerate}
\end{enumerate}

\item \textbf{The subshift $X^M$ is not $4$-block-gluing.}

We consider the two half-plane patterns $p$ and $q$ on respective supports 
$\mathbb{Z}_{-} \times \mathbb{Z}$ and 
$\mathbb{Z}_{+} \times \mathbb{Z}$ such that
for all $j \le 0$, if $-j \equiv 0,1 [4]$,
then for all $\textbf{u} \in C_j$, $p_{\textbf{u}}$ is $\begin{tikzpicture}[scale=0.3,baseline=0.4mm]
\fill[gray!85] (0,0) rectangle (1,1);
\draw (0,0) rectangle (1,1);
\end{tikzpicture}$, else for all $\textbf{u} \in C_j$, it is $\begin{tikzpicture}[scale=0.3,baseline=0.4mm]
\draw (0,0) rectangle (1,1);
\end{tikzpicture}$, and $q$ is obtained from $p$ 
by symmetry. It is easy to see that these patterns are globally admissible. We leave 4 columns between $p$ and $q$ (see \autoref{minimal-4-gluing})
To ensure that the dominant positions in columns 0 and 5, which are not isolated, have private neighbours, every cell of the four middle columns needs to be $\begin{tikzpicture}[scale=0.3,baseline=0.4mm]
\draw (0,0) rectangle (1,1);
\end{tikzpicture}$, as in \autoref{minimal-4-gluing}. This filling implies that the cells in column 2 and 3, which are not dominant, are also not dominated. This shows that the subshift is not 4 block gluing.

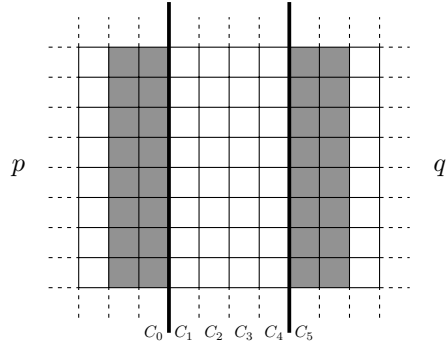
\begin{figure}[h!]

\[\begin{tikzpicture}[scale=0.4]
\fill[gray!85] (1,0) rectangle (3,8);
\fill[gray!85] (7,0) rectangle (9,8);

\foreach \x in {0,...,2} {\draw[mydashed] (\x, -1) -- (\x, 0);}
\foreach \x in {0,...,2} {\draw[mydashed] (\x, 8) -- (\x, 9);}
\foreach \x in {8,...,10} {\draw[mydashed] (\x, -1) -- (\x, 0);}
\foreach \x in {8,...,10} {\draw[mydashed] (\x, 8) -- (\x, 9);}

\foreach \x in {0,...,8} {\draw[mydashed] (-1,\x) -- (0,\x);}
\foreach \x in {0,...,8} {\draw[mydashed] (10,\x) -- (11,\x);}
\draw (0,0) grid (10,8); 

\node at (-2,4) {$p$};
\node at (12,4) {$q$};
\foreach \x in {4,...,6} {\draw[mydashed] (\x, -1) -- (\x, 0);}
\foreach \x in {4,...,6} {\draw[mydashed] (\x, 8) -- (\x, 9);}

\node[scale=0.625] at (2.5,-1.5) {$C_0$};
\node[scale=0.625] at (3.5,-1.5) {$C_1$};
\node[scale=0.625] at (4.5,-1.5) {$C_2$};
\node[scale=0.625] at (5.5,-1.5) {$C_3$};
\node[scale=0.625] at (6.5,-1.5) {$C_4$};
\node[scale=0.625] at (7.5,-1.5) {$C_5$};

\draw[line width =0.5mm] (3,-1.5) -- (3,9.5);
\draw[line width =0.5mm] (7,-1.5) -- (7,9.5);
\end{tikzpicture}\]
	
\caption{Illustration of the fact that $X^M$
is not $4$-block-gluing: when attempting to glue $p$ and $q$, ensuring the existence of private neighbours according to the rules of $X^M$ 
forces the presence of undominated positions 
(coloured in light grey). It is also a counter-example for $X^{MT}$ being 4-block-gluing.}
\label{minimal-4-gluing}
\end{figure}
As a consequence $c(X^M) >4$. Since 
it is $5$-block-gluing, $c(X^M) = 5$.
\end{itemize}
\end{proof}

\begin{thm}
The minimal total domination subshift is block gluing and $c(X^{MT}) = 5$.
\label{total-minimum-gluing-th}
\end{thm}

\noindent \textbf{Idea of the proof:}
\textit{We follow the same 
scheme as in the proof of Theorem~\ref{minimal-gluing-th}, 
except that we have to take into account 
the variations in the definition of the 
subshift $X^{MT}$. For the sake 
of readability, we reproduce the structure 
of the proof.}

\begin{proof}

\begin{figure}[h!]

\centering
\begin{tikzpicture}[scale=0.3]

%TODO Rule iii does not appear in the illustration => works it out later

\fill[gray!85] (-1,2) rectangle (1,3);
\fill[gray!85] (-1,0) rectangle (1,1);
\fill[gray!85] (-1,5) rectangle (1,6);
\fill[gray!85] (-1,7) rectangle (1,8);
\fill[gray!85] (8,4) rectangle (9,5);
\fill[gray!85] (9,1) rectangle (11,2);
\fill[gray!85] (10,3) rectangle (11,4);
\fill[gray!85] (9,6) rectangle (11,7);

\foreach \x in {-1,...,2} {\draw[mydashed] (\x, -1) -- (\x, -2);}
\foreach \x in {-1,...,2} {\draw[mydashed] (\x, 9) -- (\x, 10);}
\foreach \x in {8,...,11} {\draw[mydashed] (\x, -1) -- (\x, -2);}
\foreach \x in {8,...,11} {\draw[mydashed] (\x, 9) -- (\x, 10);}

\foreach \y in {-1,...,9} {\draw[mydashed] (-1,\y) -- (-2,\y);}
\foreach \y in {-1,...,9} {\draw[mydashed] (11,\y) -- (12,\y);}

\draw (-1,-1) grid (2,9);
\draw (8,-1) grid (11,9);
\node at (-3,4) {$p$};
\node at (13,4) {$q$};

\node[scale=0.625] at (2.6,-3) {$C_1$};
\node[scale=0.625] at (1.6,-3) {$C_0$};
\node[scale=0.625] at (8.95,-3) {$C_{k+1}$};
\node[scale=0.625] at (7.7,-3) {$C_k$};
\node[scale=0.7] at (5,-3) {$\hdots$};

\draw[line width =0.5mm] (2,-2.5) -- (2,10.5);
\draw[line width =0.5mm] (8,-2.5) -- (8,10.5);

\node[scale=1.25] at (5,-4.5) {$(0)$};

\begin{scope}[xshift=20cm]

\fill[gray!85] (-1,2) rectangle (1,3);
\fill[gray!85] (-1,0) rectangle (1,1);
\fill[gray!85] (-1,5) rectangle (1,6);
\fill[gray!85] (-1,7) rectangle (1,8);
\fill[gray!85] (8,4) rectangle (9,5);
\fill[gray!85] (9,1) rectangle (11,2);
\fill[gray!85] (10,3) rectangle (11,4);
\fill[gray!85] (9,6) rectangle (11,7);

\fill[gray!85] (2,1) rectangle (4,2);
\fill[gray!85] (2,3) rectangle (3,5);
\fill[gray!85] (2,6) rectangle (4,7);

\fill[gray!85] (7,4) rectangle (8,5);
\fill[gray!85] (6,0) rectangle (8,1);
\fill[gray!85] (6,2) rectangle (8,3);
\fill[gray!85] (6,7) rectangle (8,8);

\foreach \x in {-1,...,2} {\draw[mydashed] (\x, -1) -- (\x,-2);}
\foreach \x in {-1,...,2} {\draw[mydashed] (\x, 9) -- (\x, 10);}
\foreach \x in {8,...,11} {\draw[mydashed] (\x, -1) -- (\x, -2);}
\foreach \x in {8,...,11} {\draw[mydashed] (\x, 9) -- (\x, 10);}

\foreach \x in {3,4,5,6,7} {\draw[mydashed] (\x, -1) -- (\x, -2);}
\foreach \x in {3,4,5,6,7} {\draw[mydashed] (\x, 9) -- (\x, 10);}

\foreach \y in {-1,...,9} {\draw[mydashed] (-1,\y) -- (-2,\y);}
\foreach \y in {-1,...,9} {\draw[mydashed] (11,\y) -- (12,\y);}
\draw (-1,-1) grid (5,9);
\draw (6,-1) grid (11,9);

\draw[line width =0.5mm] (2,-2.5) -- (2,10.5);
\draw[line width =0.5mm] (8,-2.5) -- (8,10.5);

\node[scale=0.7] at (6,-3) {$C_{k-2}$};
\node[scale=1.25] at (5,-4.5) {$(1)$};

\draw[white, line width=0.5mm] (7.03,-1) -- (7.91,-1);
\draw[white, line width=0.5mm] (2.09,-1) -- (5,-1);
\draw[white, line width=0.5mm] (3,-2) -- (3,-0.1); 
\draw[white, line width=0.5mm] (4,-2) -- (4,-0.1);
\draw[white, line width=0.5mm] (5,-2) -- (5,-0.1);

\draw[white, line width=0.5mm] (7.03,9) -- (7.91,9);
\draw[white, line width=0.5mm] (5.03,9) -- (5.96,9);
\draw[white, line width=0.5mm] (2.09,9) -- (5,9); 
\draw[white, line width=0.5mm] (3,8.1) -- (3,10); 
\draw[white, line width=0.5mm] (4,8.1) -- (4,10);
\draw[white, line width=0.5mm] (5,8.1) -- (5,10);

\draw[mydashed] (3, -2) -- (3,0);
\draw[mydashed] (4,-2) -- (4,0);
\draw[mydashed] (5,-2) -- (5,0);

\draw[mydashed] (3,8) -- (3,10);
\draw[mydashed] (4,8) -- (4,10);
\draw[mydashed] (5,8) -- (5,10);

\end{scope}

\begin{scope}[yshift=-16.5cm]

\fill[gray!85] (-1,2) rectangle (1,3);
\fill[gray!85] (-1,0) rectangle (1,1);
\fill[gray!85] (-1,5) rectangle (1,6);
\fill[gray!85] (-1,7) rectangle (1,8);
\fill[gray!85] (8,4) rectangle (9,5);
\fill[gray!85] (9,1) rectangle (11,2);
\fill[gray!85] (10,3) rectangle (11,4);
\fill[gray!85] (9,6) rectangle (11,7);

\fill[gray!85] (5,2) rectangle (6,6);
\fill[gray!85] (2,1) rectangle (4,2);
\fill[gray!85] (2,3) rectangle (3,5);
\fill[gray!85] (2,6) rectangle (4,7);

\fill[gray!85] (7,4) rectangle (8,5);
\fill[gray!85] (6,0) rectangle (8,1);
\fill[gray!85] (6,2) rectangle (8,3);
\fill[gray!85] (6,7) rectangle (8,8);

\foreach \x in {-1,...,11} {\draw[mydashed] (\x, -2) -- (\x, -1);}
\foreach \x in {-1,...,11} {\draw[mydashed] (\x, 9) -- (\x, 10);}

\foreach \y in {-1,...,9} {\draw[mydashed] (-1,\y) -- (-2,\y);}
\foreach \y in {-1,...,9} {\draw[mydashed] (11,\y) -- (12,\y);}
\draw (-1,-1) grid (11,9);

\draw[line width =0.5mm] (2,-2.5) -- (2,10.5);
\draw[line width =0.5mm] (8,-2.5) -- (8,10.5);

\node[scale=1.25] at (5,-4) {$(2)$};

\draw[white, line width=0.5mm] (7.03,-1) -- (7.91,-1);
\draw[white, line width=0.5mm] (2.09,-1) -- (4.96,-1); 
\draw[white, line width=0.5mm] (3,-2) -- (3,-0.1); 
\draw[white, line width=0.5mm] (4,-2) -- (4,-0.1);
\draw[white, line width=0.5mm] (5,-2) -- (5,-1.1);

\draw[white, line width=0.5mm] (7.03,9) -- (7.91,9);
\draw[white, line width=0.5mm] (5.03,9) -- (5.96,9);
\draw[white, line width=0.5mm] (2.09,9) -- (5.96,9); 
\draw[white, line width=0.5mm] (3,8.1) -- (3,10); 
\draw[white, line width=0.5mm] (4,8.1) -- (4,10);
\draw[white, line width=0.5mm] (5,8.1) -- (5,10);

\draw[mydashed] (3, -2) -- (3,0);
\draw[mydashed] (4,-2) -- (4,0);
\draw[mydashed] (5,-2) -- (5,-1);

\draw[mydashed] (3,8) -- (3,10);
\draw[mydashed] (4,8) -- (4,10);
\draw[mydashed] (5,8) -- (5,10);

\draw (4.96,-1) -- (6,-1);
\draw (5.96,9) -- (6,9);

\end{scope}

\begin{scope}[yshift=-16.5cm,xshift=20cm]

\fill[gray!85] (-1,2) rectangle (1,3);
\fill[gray!85] (-1,0) rectangle (1,1);
\fill[gray!85] (-1,5) rectangle (1,6);
\fill[gray!85] (-1,7) rectangle (1,8);
\fill[gray!85] (8,4) rectangle (9,5);
\fill[gray!85] (9,1) rectangle (11,2);
\fill[gray!85] (10,3) rectangle (11,4);
\fill[gray!85] (9,6) rectangle (11,7);

\fill[gray!85] (5,2) rectangle (6,6);
\fill[gray!85] (2,1) rectangle (4,2);
\fill[gray!85] (2,3) rectangle (3,5);
\fill[gray!85] (2,6) rectangle (4,7);

\fill[gray!85] (7,4) rectangle (8,5);
\fill[gray!85] (6,0) rectangle (8,1);
\fill[gray!85] (6,2) rectangle (8,3);
\fill[gray!85] (6,7) rectangle (8,8);

\foreach \x in {-1,...,11} {\draw[mydashed] (\x, -2) -- (\x, -1);}
\foreach \x in {-1,...,11} {\draw[mydashed] (\x, 9) -- (\x, 10);}

\foreach \y in {-1,...,9} {\draw[mydashed] (-1,\y) -- (-2,\y);}
\foreach \y in {-1,...,9} {\draw[mydashed] (11,\y) -- (12,\y);}
\draw (-1,-1) grid (11,9);

\draw[line width =0.5mm] (2,-2.5) -- (2,10.5);
\draw[line width =0.5mm] (8,-2.5) -- (8,10.5);

\node[scale=1.25] at (5,-4) {$(3)$};

\draw[white, line width=0.5mm] (7.03,-1) -- (7.91,-1);
\draw[white, line width=0.5mm] (2.09,-1) -- (4.96,-1); 
\draw[white, line width=0.5mm] (3,-2) -- (3,-0.1); 
\draw[white, line width=0.5mm] (4,-2) -- (4,-0.1);
\draw[white, line width=0.5mm] (5,-2) -- (5,-1.1);

\draw[white, line width=0.5mm] (7.03,9) -- (7.91,9);
\draw[white, line width=0.5mm] (5.03,9) -- (5.96,9);
\draw[white, line width=0.5mm] (2.09,9) -- (5.96,9); 
\draw[white, line width=0.5mm] (3,8.1) -- (3,10); 
\draw[white, line width=0.5mm] (4,8.1) -- (4,10);
\draw[white, line width=0.5mm] (5,8.1) -- (5,10);

\draw[mydashed] (3, -2) -- (3,0);
\draw[mydashed] (4,-2) -- (4,0);
\draw[mydashed] (5,-2) -- (5,-1);

\draw[mydashed] (3,8) -- (3,10);
\draw[mydashed] (4,8) -- (4,10);
\draw[mydashed] (5,8) -- (5,10);

\draw (4.96,-1) -- (6,-1);
\draw (5.96,9) -- (6,9);

\end{scope}
\end{tikzpicture}
\caption{\label{figure.completing.algorithm.1} Illustration of the algorithm for 
filling the intermediate columns between two half-plane patterns $p$ and $q$ for 
the minimal total domination subshift. In 
the last step, the position $\textbf{u}_0$ 
is the bottommost represented positon 
of the central column, and the central column is coloured 
with a possible colouring. We chose $k=6$, still the proof works with $k=5$.}
\end{figure}
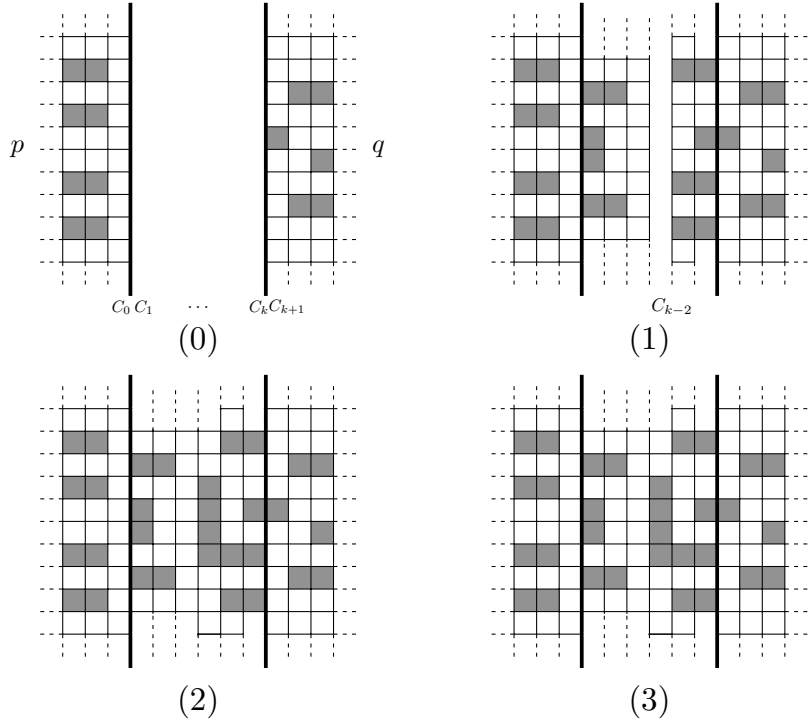

\leavevmode

\begin{itemize} 
\item \textbf{Filling the intermediate columns 
between two half-plane patterns.} 

We provide here an algorithm to fill
these columns between two patterns $p$ and $q$ respectively on 
$\mathbb{Z}_{-} \times \mathbb{Z}$ and 
$\mathbb{Z}_{+} \times \mathbb{Z}$ into 
a configuration $x \in X^{MT}$: 

\begin{enumerate}
\vspace{-0,1cm}
\setlength{\parskip}{0pt}
\setlength{\itemsep}{3pt}
\item \textbf{Filling the intermediate columns, from $C_1$ to $C_{k-3}$, then $C_k,C_{k-1}$.} 
Successively, for all $j$ 
from $1$ to $k-3$, we determine the column 
$C_j$ according to the following rule: 
for all $\textbf{u} \in C_j$, $x_{\textbf{u}}$ 
is $\begin{tikzpicture}[scale=0.3,baseline=0.4mm]
\fill[gray!85] (0,0) rectangle (1,1);
\draw (0,0) rectangle (1,1);
\end{tikzpicture}$ when
$x_{\textbf{u}-(1,1)}$, $x_{\textbf{u}-(1,-1)}$ 
and $x_{\textbf{u}-(2,0)}$ are $\begin{tikzpicture}[scale=0.3,baseline=0.4mm]
\draw (0,0) rectangle (1,1);
\end{tikzpicture}$ (the difference with the 
proof of Theorem~\ref{minimal-gluing-th} is 
that the symbol $x_{\textbf{u}-(1,0)}$ is 
not imposed). Else, $x_{\textbf{u}}$ is set to $\begin{tikzpicture}[scale=0.3,baseline=0.4mm]
\draw (0,0) rectangle (1,1);
\end{tikzpicture}$. This rule 
is illustrated in Figure~\ref{figure.local.rules.completion}:

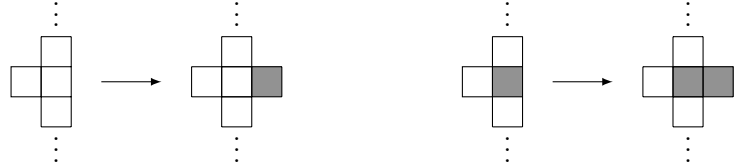
\begin{figure}[h!]
\[\begin{tikzpicture}[scale=0.4]
\draw (0,0) grid (1,3);
\draw (-1,1) rectangle (0,2);
\node at (0.5,-0.5) {$\vdots$};
\node at (0.5,4) {$\vdots$};

\draw[-latex] (2,1.5) -- (4,1.5);

\begin{scope}[xshift=6cm]
\node at (0.5,-0.5) {$\vdots$};
\node at (0.5,4) {$\vdots$};
\fill[gray!85] (1,1) rectangle (2,2);
\draw (0,0) grid (1,3);
\draw (-1,1) grid (2,2);
\end{scope}

\begin{scope}[xshift=15cm]
\fill[gray!85] (0,1) rectangle (1,2);
\draw (0,0) grid (1,3);
\draw (-1,1) rectangle (0,2);
\node at (0.5,-0.5) {$\vdots$};
\node at (0.5,4) {$\vdots$};

\draw[-latex] (2,1.5) -- (4,1.5);

\begin{scope}[xshift=6cm]
\node at (0.5,-0.5) {$\vdots$};
\node at (0.5,4) {$\vdots$};
\fill[gray!85] (0,1) rectangle (1,2);
\fill[gray!85] (1,1) rectangle (2,2);
\draw (0,0) grid (1,3);
\draw (-1,1) grid (2,2);
\end{scope}
\end{scope}

\end{tikzpicture}\]
\caption{\label{figure.local.rules.completion} 
Illustration of the local rules for the completion 
algorithm for the intermediate columns.}
\end{figure}

For $j=k$ and then $j=(k-1)$, 
we determine $x$ on any position $x_{\textbf{u}}$ 
for $\textbf{u} \in C_j$ by applying 
a symmetrical rule: $x_{\textbf{u}}$ 
is $\begin{tikzpicture}[scale=0.3,baseline=0.4mm]
\fill[gray!85] (0,0) rectangle (1,1);
\draw (0,0) rectangle (1,1);
\end{tikzpicture}$ when 
$x_{\textbf{u}+(1,1)}$, $x_{\textbf{u}+(1,-1)}$ 
and $x_{\textbf{u}+(2,0)}$ are $\begin{tikzpicture}[scale=0.3,baseline=0.4mm]
\draw (0,0) rectangle (1,1);
\end{tikzpicture}$. Else it is 
$\begin{tikzpicture}[scale=0.3,baseline=0.4mm]
\draw (0,0) rectangle (1,1);
\end{tikzpicture}$.

%\item Stage 2: Do the same for column 2 and 4: put an 'x' in a cell of column 2 (resp. 4) if its neighbour in column 1 (resp. 5) was not yet dominated;
%\item Stage 2': If $k > 5$, repeat Stage 2 for columns $3, \cdots, k-3$. For column $i$, the neighbour we consider is the one in column $i-1$;
\item \textbf{The central column ($j=k-2)$.} 

We then 
determine $x$ on the central column $C_{k-2}$. For all $\textbf{u} \in C_{k-2}$, $x_{\textbf{u}}$ 
is $\begin{tikzpicture}[scale=0.3,baseline=0.4mm]
\fill[gray!85] (0,0) rectangle (1,1);
\draw (0,0) rectangle (1,1);
\end{tikzpicture}$ when 
$x_{\textbf{u}+(1,1)}$, $x_{\textbf{u}+(1,-1)}$ 
and $x_{\textbf{u}+(2,0)}$ are equal to 
$\begin{tikzpicture}[scale=0.3,baseline=0.4mm]
\draw (0,0) rectangle (1,1);
\end{tikzpicture}$, 
or when
$x_{\textbf{u}-(1,1)}$, $x_{\textbf{u}-(1,-1)}$ 
and $x_{\textbf{u}-(2,0)}$ are equal to 
$\begin{tikzpicture}[scale=0.3,baseline=0.4mm]
\draw (0,0) rectangle (1,1);
\end{tikzpicture}$. Else it is 
$\begin{tikzpicture}[scale=0.3,baseline=0.4mm]
\draw (0,0) rectangle (1,1);
\end{tikzpicture}$.

\item \textbf{Eliminating minimality and total domination errors 
in the central column.} 

Choose any position $\textbf{u}_0 \in C_{k-2}$. From this position upwards, check 
for all positions if they are dominated.
If this is not the case, then change 
the symbol on $\textbf{u}+(0,1)$ into 
$\begin{tikzpicture}[scale=0.3,baseline=0.4mm]
\fill[gray!85] (0,0) rectangle (1,1);
\draw (0,0) rectangle (1,1);
\end{tikzpicture}$.
After $\textbf{u}_0$ has been processed do the same symmetrically (change the symbol in $\textbf{u}_0-(0,1)$ when $\textbf{u}$ is not dominated) in parallel downwards, beginning from $\textbf{u}_0-(0,1)$.

%\item Stage 3: For each cell in column $k-2$, put an 'x' if its neighbour in column $k-3$ or $k-1$ is not dominated;
%\item Stage 4: From cell $(k-2, 0)$ downwards, if a cell is not dominated, then put an 'x' inside. Do the same from $(k-2, 1)$ upwards.
\end{enumerate}

See an illustration of this algorithm on 
Figure~\ref{figure.completing.algorithm}.

\item \textbf{The obtained configuration 
is in $X^{MT}$.}

We have to check that the local 
rules of the minimal total 
domination subshift 
are verified over all the constructed configuration $x$.
\vspace{5cm}
\begin{enumerate}

\item \textbf{Local rules are verified inside 
the half-planes.}

Same as the corresponding point 
in the proof of Theorem~\ref{minimal-gluing-th}.

\item \textbf{Every position outside $\text{supp}(p) \cup \text{supp}(q)$ is dominated.}

Same as the corresponding point 
in the proof of Theorem~\ref{minimal-gluing-th} for the positions 
outside the central column $C_{k-2}$. 
In this column, let us assume 
that a position $\textbf{u}$ 
above $\vec{u}_0$ (without loss of 
generality) is not dominated. Then 
the last step of the algorithm, when 
scanning this position, would have
changed the symbol on position $\textbf{u}+(0,1)$, which is a contradiction.

In particular, no dominant positions are 
isolated.

\item \textbf{Every dominant position outside $\text{supp}(p) \cup \text{supp}(q)$ has a private neighbour.}

\begin{enumerate}
\item[(a+b)] Outside the central column, the 
proof is similar to the corresponding points 
in the proof of Theorem~\ref{minimal-gluing-th}.
\item[(c)] In the column $C_{k-2}$, the 
dominant positions added in Step 2 necessarily
have a private neighbour in column $C_{k-3}$ 
or $C_{k-1}$. Let us take a dominant position $\textbf{u}$, assumed without loss of generality to be above $\textbf{u}_0$, which was added in the last step. This implies that $\textbf{u}-(0,1)$ was not dominated when 
the algorithm checked this position. As a consequence, it is a private neighbour for $\textbf{u}$.
\end{enumerate}
\end{enumerate}

\item 

\textbf{The subshift $X^{MT}$ is not $4$-block-gluing.}

Let us consider the patterns $p$ and $q$ 
defined in the corresponding point in the proof of Theorem~\ref{minimal-gluing-th}.
It is easy to see that these two patterns are also globally admissible in the subshift $X^{MT}$. We only have to check that in the constructed
configuration, no dominant positions are 
isolated, which is straightforward.

Using the same arguments as the ones for the minimal domination case, it is easy to see that any configuration in $\mathcal{A}^{\mathbb{Z}^2}$ where $p$ and $q$ are glued at distance 4 contains some forbidden patterns.

As a consequence $c(X^{MT}) >4$. Since 
it is $5$-block-gluing, $c(X^{MT}) = 5$.
\end{itemize}
\end{proof}

As a direct consequence of Theorem~\ref{theorem.computability.entropy}: 

\begin{thm}
The numbers $\nu_D$ and $\nu_T$ 
are computable with rate $n \mapsto 2^{n^2}$. 
The numbers $\nu_M$ and $\nu_{MT}$
are computable with rate $n \mapsto 2^{5n^2}$.
\end{thm}

\section{\label{section.computer} Computing bounds for the growth rate}

Although the algorithm presented in Section~\ref{subsection.computability} provides a way to compute the growth 
rates of various dominating sets of the grids $G_{n,m}$, 
it is not efficient enough for practical use on a computer.
In this section, we use other tools which make it possible to obtain bounds for the growth rates, although 
with no guarantee on their precision. These bounds are obtained using computer resources, by running a C++ program made for the occasion. 
The technique relies on, for a fixed $m$, assimilating the
dominating sets of $G_{n,m}$ 
to patterns of a unidimensional subshift 
of finite type, whose entropy is known 
to be computable through linear algebra 
computing.

This method is well known. It was for instance used, along with other techniques, to solve the problem of finding the minimum size of a dominating (see \cite{rao}), 2-dominating and Roman dominating (see \cite{talon}) set of a grid of arbitrary size. These papers provide an alternate explanation without relating it to the theory of SFTs. For instance, Section 2 in \cite{talon} uses the same technique as the one we use here, but in the (min,+)-algebra. However, in their paper there are no such things as lower or upper bounds we investigate here: they only enumerate sets which are exactly 2-dominating or Roman dominating. Since they are interested in finding the minimum size of such a set, they can apply some optimisations to avoid enumerating some sets which cannot be of minimum size. Since we want to count \emph{all} the different dominating sets, these optimisations do not apply here.

\subsection{Relating $(D_{n,m})$ to the entropy of some unidimensional SFT}

\subsubsection{\label{subsection.nearest.neighbour} Nearest-neighbour unidimensional subshifts of finite type}

In this section $\mathcal{A}=(a_1,...,a_k)$ 
is a finite set, 
and $X$ a unidimensional subshift of finite type
on alphabet $\mathcal{A}$. Let us 
denote by $(e_1,...,e_k)$ the canonical basis 
of $\mathbb{R}^k$.

\begin{definition}
The subshift $X$ is said to be 
\textbf{nearest neighbour} when it is defined 
by forbidding a set of patterns on support 
$\{0,1\}$.
\end{definition}

\begin{definition}
The \textbf{adjacency matrix} of $X$ is the matrix $M \in \mathcal{M}_k (\mathbb{R})$ 
such that $M[e_i,e_j]=1$ if the 
pattern $a_i a_j$ is not forbidden, or 0 otherwise. 
\end{definition}

The following is well known: 

\begin{proposition}
Let $||.||$ be any matricial norm. The entropy of $X$ is equal to the spectral radius of $M$: 
\[h(X)= \log_2{\lim_n ||M^n||^{1/n}}.\]
\end{proposition}

\subsubsection{Unidimensional versions of the domination subshifts}

We define here the unidimensional versions of 
the domination subshifts defined in Section~\ref{section.examples}. We use them to describe and prove the method we use to obtain the bounds on the growth rates. The first one ($X^{D, m}$) is used to obtain the lower bound, whereas we use the second one ($X_*^{D,m}$) for the upper bound.

\begin{notation}
Let us fix some integer $m \ge 1$. 
We denote by $X^{D,m}$ the 
undimensional subshift on alphabet $\mathcal{A}_0 ^n$ such that a configuration $x$ 
is in $X^{D,m}$ if and only if 
the set of positions $(j,k) \in \mathbb{Z} \times \llbracket 1,m\rrbracket$ such that the symbol $(x_j)_k$ 
is grey forms a dominating set of the 
grid $\mathbb{Z} \times \llbracket 1,m\rrbracket$.
\end{notation}

With similar arguments as in the proofs of Lemma~\ref{lemma.comparison.minimal} and Lemma~\ref{lemma.comparison.minimal.total},
we get that when $m$ is fixed and $n$ grows 
to infinity:
\[D_{n,m} = 2^{h(X^{D,m}) \cdot n + o(n)}\]

\begin{notation}
For all $m \ge 3$, we also 
denote by $X_{*}^{D,m}$ the 
undimensional subshift on alphabet $\mathcal{A}_0 ^n$ such that a configuration $x$ 
is in $X_{*}^{D,m}$ if and only if 
the set of positions $(j,k) \in \mathbb{Z} \times \llbracket 2,m-1\rrbracket$ such that the symbol $(x_j)_k$ 
is grey forms a dominating set of the 
grid $\mathbb{Z} \times \llbracket 2,m-1\rrbracket$.
\end{notation}

\subsubsection{Recoding into nearest-neighbour 
subshifts}

Let us set $\mathcal{A}_1 = \left\{ 
\begin{tikzpicture}[scale=0.3]
\draw (0,0) rectangle (1,1);
\end{tikzpicture},
\begin{tikzpicture}[scale=0.3]
\draw[fill=gray!30] (0,0) rectangle (1,1);
\end{tikzpicture}, \begin{tikzpicture}[scale=0.3]
\draw[fill=gray!85] (0,0) rectangle (1,1);
\end{tikzpicture}\right\}$, 
and let us consider the map 
$\varphi : \left(\mathcal{A}^m_0\right)^{\mathbb{Z}} \rightarrow \left(\mathcal{A}^m_1\right)^{\mathbb{Z}}$ 
that acts on configurations of $\left(\mathcal{A}^n_0\right)^{\mathbb{Z}}$ by changing the $i^\text{th}$ 
symbol of any position $j \in \mathbb{Z}$ 
into $\begin{tikzpicture}[scale=0.3]
\draw[fill=gray!30] (0,0) rectangle (1,1);
\end{tikzpicture}$ whenever it is not dominant and 
dominated by an element of $C_{j-1} \bigcap \left( \mathbb{Z} \times \llbracket 1,m\rrbracket \right)$ or $C_{j} \bigcap \left( \mathbb{Z} \times \llbracket 1,m\rrbracket \right)$. Informally, from lightest to darkest they stand for an undominated cell (which is not dominant), a dominated cell which is not dominant and a dominant cell. This is illustrated in Figure~\ref{figure.illustration.conjugation}. The nearest-neighbour property makes it possible to count the dominating sets without enumerating them fully: it is enough to store a small number of the latest columns, proceeding from left to right in the grid.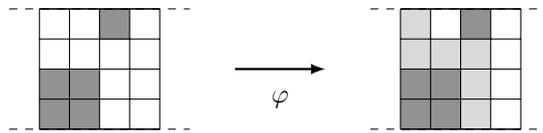
\begin{figure}[h!]
\[\begin{tikzpicture}[scale=0.4]

\fill[gray!85] (1,0) rectangle (3,2);
\fill[gray!85] (3,3) rectangle (4,4);
\draw[dashed] (0,0) -- (6,0);
\draw[dashed] (0,4) -- (6,4);
\draw (1,0) -- (5,0);
\draw (1,4) -- (5,4);
\draw (1,0) grid (5,4);

\draw[line width=0.3mm,-latex] (7.5,2) -- (10.5,2);
\node at (9,1) {$\varphi$};

\begin{scope}[xshift=12cm]

\fill[gray!30] (3,0) rectangle (4,3);
\fill[gray!30] (1,2) rectangle (3,3);
\fill[gray!30] (1,3) rectangle (2,4);

\fill[gray!85] (1,0) rectangle (3,2);
\fill[gray!85] (3,3) rectangle (4,4);
\draw[dashed] (0,0) -- (6,0);
\draw[dashed] (0,4) -- (6,4);
\draw (1,0) -- (5,0);
\draw (1,4) -- (5,4);
\draw (1,0) grid (5,4);

%\fill[gray!85] (2,0) rectangle (3,2);
%\fill[gray!85] (3,3) rectangle (4,4);
%\draw[dashed] (0,0) -- (6,0);
%\draw[dashed] (0,4) -- (6,4);
%\draw (1,0) -- (5,0);
%\draw (1,4) -- (5,4);
%\draw (2,0) grid (5,4);
%\draw[dashed] (1,0) grid (2,4);
\end{scope}
\end{tikzpicture}\]
\caption{\label{figure.illustration.conjugation} Illustration of the map recoding 
$X^{D,m}$ into a nearest-neighbour SFT.}
\end{figure}

Since $\varphi$ commutes with the shift action 
and is invertible, 

\[h(X^{D,m}) = h(\varphi(X^{D,m})).\]

Moreover, the subshift $X^{D,m}$ has the nearest-neighbour property.

\subsection{Numerical approximations}

We use these last equalities to prove the following:
% TODO ^ which equalities ? Several of them???
\begin{thm}[domination]
The following inequalities hold $ 1.950022198 \le \nu_D \le 1.959201684$.
\end{thm}

\begin{proof}\leavevmode

\begin{itemize} \item \textbf{Lower bound:}

\begin{enumerate}
\item For all integers $n,m_1,m_2$, $D_{n,m_1+m_2} \ge D_{n,m_1} \cdot D_{n,m_2}.$

Indeed, let us consider two
sets dominating respectively $G_{n,m_1}$ and $G_{n,m_2}$. 
By gluing the first one on the top of 
the second one, we obtain a dominating set of 
$G_{n,m_1+m_2}$. This is true because any position in this grid 
is either in the copy of the grid $G_{n,m_1}$
and thus dominated by an element in this grid, 
or in the copy of $G_{n, m_2}$. Since this construction is invertible, we obtain the announced inequality.

\item As a consequence, for all $k \ge 0$,
\[D_{n,18k} \ge D_{n,18}^{k} = 2^{h(X^{D,18}) \cdot kn + k \cdot o(n)},\]
where the function $o(n)$ is related to the fact that we used $18$ lines.
This implies that
\[\lim_{n,m} \frac{\log_2 (D_{n,m})}{nm} 
= \lim_{n,k} \frac{\log_2 (D_{n,18k})}{18nk} \ge h(X^{D,18}).\]
\item This number is equal to 
$h(\varphi(X^{D,18}))$, which is computed using Section~\ref{subsection.nearest.neighbour}. The lower bound follows.
\end{enumerate}

\item \textbf{Upper bound:}

\begin{enumerate}
\item 
For all $n,m$, let us denote by $D^{*}_{n,m}$ 
the number of sets of vertices of $G_{n,m}$ 
which dominate the middle $m-2$ lines (i.e. cells of the first and last lines might not be dominated). We have a direct inequality \[D_{n,m} \le D^{*}_{n,m}.\]
\item For a reason similar as the one in the first point 
of the proof of the lower bound, for all $m_1,m_2$, $D^{*}_{n,m_1+m_2} \le D^{*}_{n,m_1} \cdot D^{*}_{n,m_2}$.
\item For all $k \ge 0$ and $n \ge 0$, 
\[D^{*}_{n,18k} \le (D^{*}_{n,18})^{k} = 2^{h(X_{*}^{D,18}) \cdot km + k \cdot o(n)}.\] 
As a consequence 
\[h(X) \le h(X_{*}^{D,18}).\]
With the same method as for the lower bound, 
we obtain the upper bound.
\end{enumerate}

\end{itemize}
\end{proof}

\begin{remark}
With further numerical manipulations, we notice that the lower bound and the upper bound seem to get closer to each other rather slowly. To speed up the convergence, we had the idea of using the sequences of ratios $h(X^{D,m+1})/h(X^{D,m})$ and 
$h(X_{*}^{D,m+1})/h(X_{*}^{D,m})$. This seems to offer a much better convergence speed. Indeed, for both sequences, from $m=11$ on, the ratio seem to be stabilised around $1.954751195$.
\label{rk-convergence}
\end{remark}

Using the same method, we provide bounds for some other problems. For the total domination, we can make the computations until $m = 17$. However, for the other problems (the ones with the minimality constraint) the number of patterns we enumerate grows exponentially at a much faster rate than for the domination problem, thus the bounds are less good. We cannot go further than around $m=10$ for these problems.

\begin{thm}[total domination]
$ 1.904220376\le \nu_T \le 1.923434191$.
\end{thm}

\begin{remark}
As in \autoref{rk-convergence}, the ratios $h(X^{T,m+1})/h(X^{T,m})$ and $h(X_{*}^{T,m+1})/h(X_{*}^{T,m})$) offer a much better convergence speed. Indeed, from $m=10$ they seem to stabilise, both around $1.915316$.
\end{remark}

\begin{thm}[minimal domination]
$1.315870482 \le \nu_M \le 1.550332154$.
\end{thm}

\begin{thm}[minimal total domination]
$1.275805204 \le \nu_{MT} \le 1.524476040$.
\end{thm}

\end{document}